\documentclass[12pt]{article}

\usepackage{hyperref}
\usepackage[sort&compress,numbers]{natbib}
\usepackage{doi}
\usepackage[margin=0.8in]{geometry}

\usepackage{authblk}

\usepackage{amsmath, amssymb, amsthm, mathtools, mathrsfs}
\usepackage{cmll}

\newcommand*{\vcbox}[1]{\begingroup
\setbox0=\hbox{#1}\parbox{\wd0}{\box0}\endgroup}

\def\Me{\mathcal M}
\def\Ce{\mathcal C}
\def\Pe{\mathcal P}
\def\Te{\mathcal T}

\def\TT{\mathbb T}
\def\Fe{\mathcal F}
\def\Reg{\mathcal R}
\def\supp{\operatorname{supp}}
\def\vtl{\vartriangleleft}
\def\permut{\mathscr S}

\def\Min{\operatorname{Min}}
\def\Tr{\operatorname{Tr}}

\newtheorem{theorem}{Theorem}[section]
\newtheorem{lemma}[theorem]{Lemma}
\newtheorem{coro}[theorem]{Corollary}
\newtheorem{prop}[theorem]{Proposition}

\theoremstyle{definition}
\newtheorem{defi}[theorem]{Definition}

\theoremstyle{remark}
\newtheorem{remark}{Remark}

\newtheorem{exm}{Example}
\newtheorem*{obs}{Observation}

\title{Order structure and signalling in  higher order quantum maps}
\author{Anna Jen\v cov\'a\footnote{email: anna.jencova@mat.savba.sk}}
\affil{\small Mathematical Institute, Slovak Academy of Sciences, \v{S}tef\'anikova 49,
814 73 Bratislava, Slovakia}
\date{}

\begin{document}

\maketitle

\abstract{
We study the signalling structure of higher order quantum maps from an
order-theoretic perspective, building on the combinatorial characterization of higher
order types by Bisio and Perinotti. We have shown in a previous work that types are
represented by boolean functions called type functions, and that each such function is
characterized by a related structure poset.  We  characterize the distributive lattice generated by all
type functions with fixed indices of input and output systems -- whose elements we call regular subtypes
-- by a monotonicity condition. Unlike the set of type functions,
the lattice of regular subtypes  is closed under the one-way signalling product, moreover, it is generated by a specific family of causally ordered types.  We then study signalling relations for maps belonging to a regular subtype, showing that the no-signalling conditions  between an input and an output system are determined by a single evaluation of the corresponding 
function. For higher order types specifically, we show that all signalling relations can
be read off directly from the structure  poset via a rank
parity condition. Finally, we study relations between the structure 
poset of a type and its normal forms, that is, expressions of the type in terms of causally ordered types.
We illustrate construction of normal forms on some examples, demonstrating the possibility
that the normal form can be systematically derived from maximal chains of the poset and
signalling relations between them.}

\section{Introduction}\label{sec:intro}

Quantum channels are the fundamental objects in quantum theory. All the basic operations,  such
as state preparations, measurements and state transformations are given by quantum channels
with some specific properties. To describe manipulations of quantum channels, quantum
supermaps (or superchannels) were introduced as transformations of 'higher order', whose inputs and outputs
are quantum channels \cite{chiribella2008transforming}. Pursuing this idea further leads to establishing the hierarchy  of Higher order quantum maps, 
where each level consists of transformations that map between transformations  on lower levels, satisfying certain
admissibility conditions. In fact, all the higher order operations are themselves quantum channels
with multipartite input and output spaces, satisfying certain restrictions.
This  hierarchy  provides a  suitable framework for  description 
of very general quantum protocols and their behaviour  in an operational way.

An important instance of a higher order map is a process that transforms finite sets  of
quantum channels into another channel. A process of this type, more precisely its Choi
operator, is called a process matrix \cite{oreshkov2012quantum}.  Such a transformation may have a definite causal order, which
means that the  input channels are processed in some fixed order. In that case, it is a  quantum
comb \cite{chiribella2008quantum,chiribella2009theoretical}, consisting of a sequence of quantum channels connected through  ancillary
inputs and outputs. The picture below shows a two-slot process matrix and a corresponding
3-comb.

\begin{minipage}{0.45\textwidth}
\centering
\includegraphics[scale=0.9]{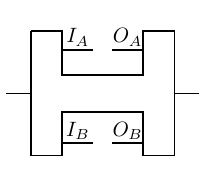}
\end{minipage}%
\begin{minipage}{0.4\textwidth}
\centering
\includegraphics[scale=0.9]{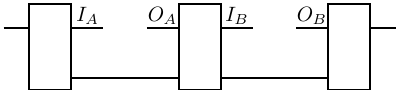}
\end{minipage}%

However, there are process matrices where the causal order is indefinite, such as the
quantum switch \cite{chiribella2013quantum} which acts on a pair of channels and produces a coherent mixture of their
composition in different orders. Such operations may give  advantages in some communication
and information processing tasks. This motivates  a deeper study of higher order
maps and their causal structure. See the
recent review \cite{taranto2025higher} for more details on higher  order quantum
processes,  related approaches, constructions, applications  and examples.

One of the basic approaches developed to study higher order maps is the theory of types by
Bisio and Perinotti \cite{perinotti2017causal, bisio2019theoretical, apadula2024nosignalling}, closely related to the approaches via
category theory \cite{kissinger2019acategorical,simmons2022higher, simmons2024acomplete}. Using the Choi representation, it was shown in
\cite{bisio2019theoretical} that the set of all channels of a given type can be
described as the set of positive operators of a fixed trace satisfying some linear
constraints. This lead to the  characterization of higher order types  by  superoperator
projections in \cite{hoffreumon2026projective,milz2024characterising}. As observed in
\cite{hoffreumon2026projective}, the projections corresponding to types form a Boolean
algebra that can be endowed with further operations of no-signalling, full signalling  and
one-way signalling
compositions, so that the structure of types can be studied by applying the algebraic
properties  of these operations. Similar constructions were found in the categorical
approach of \cite{simmons2022higher}, where the higher order theories were shown to have the
structure of a BV-logic.

The present paper is a continuation of (a part of) a  previous work \cite{jencova2024onthestructure},
inspired mainly by the approaches in \cite{perinotti2017causal, bisio2019theoretical} and
\cite{hoffreumon2026projective}. There it was pointed out that the combinatorial characterization
of types described in \cite{bisio2019theoretical,apadula2024nosignalling} leads to the identification of types
with certain boolean functions, called the type functions. The type functions naturally
live in a boolean algebra formed by (an interval in) the set of all boolean functions with
pointwise defined operations,
where the compositions analogous to those in \cite{hoffreumon2026projective} can
be conveniently introduced and have similar properties. From any type function, given a 
 list of elementary quantum systems on which the maps should act, one can construct the set of all higher order
quantum operations of the given type, and the algebraic operations correspond to
manipulations of such sets (intersections, affine combinations, tensor products, sets of
all effects, etc.) 

Apart from studying the algebraic structure of the set of all types, the type function can
be used to reveal the structure of the corresponding type. Using the M\"obius transform,
we relate to each type function a certain poset with vertices labelled by indices of the
elementary systems involved in the type. This poset (which we call the reduced structure
poset of a type function in the present work) uniquely represents the corresponding  type
function. Furthermore, this poset is a chain if and only if the corresponding type is causally ordered (that is, a quantum
comb). For a general type function, the maximal chains of the poset show the causal orders
involved in the corresponding  type and can be used to construct its 'normal form' as introduced in
\cite{hoffreumon2026projective}, expressing each type in terms of causally ordered ones.

In this work we extend these results in several ways. We first study  the set $\Te_{n,O}$ of all
type functions with fixed input and output indices. Apart from some trivial cases, this
set is not a lattice, and is strictly contained in the interval between the type functions of constant channels and the
type function of all channels with the given input/output structure. We study the distributive
lattice generated by $\Te_{n,O}$,  whose elements  we call the
regular subtypes. These are boolean functions describing sets of channels  that may be obtained
by taking affine combinations of channels that belong to one or more higher order types
with a given input/output structure. 

We show that while the set of type functions is not closed under the
causal (i.e. one-way signalling) products, the set of regular subtypes is.  We further
characterize regular subtypes  by a certain monotonicity condition.  For any regular
subtype, we prove a simple conditions that determines the no-signalling relations that must
be satisfied by any channel in the corresponding set.
We then turn to the special case of higher order
types,  where we show that the no-signalling relations can be seen from  the order
structure of the corresponding reduced structure poset. 
Finally, we show that the  normal form of a type can be constructed from a set of causally
ordered types whose cardinality is upper bounded by the number of maximal chains 
in the reduced structure poset of the type.  

The remainder of this paper is structured as follows. In Section \ref{sec:types}, we review the theory
of types of \cite{bisio2019theoretical} and the projective characterization of
\cite{hoffreumon2026projective,milz2024characterising}. In Section \ref{sec:functions}, we introduce the
type functions and the related constructions. In particular, we observe that the M\"obius
transform of a type function provides a translation between the combinatorial description
of types and the projective characterizations provided in \cite{milz2024characterising}. 
Section \ref{sec:subtypes} is devoted to characterization of regular subtypes, the signalling relations and
the normal form are studied in Section \ref{sec:signaling}. Some basic results on boolean functions are
summarized in Appendix \ref{app:fn}. More technical proofs regarding the properties of
the structure posets are given in Appendix \ref{app:structure}. In Appendix
\ref{app:normal}, we demonstrate on examples that a normal form can be constructed just by
considering the maximal chains of the reduced poset and the signalling relations. A more thorough  investigation of the relation between the reduced structure poset and normal forms of a type function  is left for future work.

\section{The type structure of higher order maps} \label{sec:types}

\subsection{Types}

 Types of higher order quantum maps and their combinatorial characterization
 were introduced in \cite{perinotti2017causal, bisio2019theoretical}. This framework
 starts with elementary types $A$, $B$, etc, representing basic quantum systems with
 underlying (finite dimensional) Hilbert spaces $\mathcal H_A$, $\mathcal H_B$, etc. 
The higher layers are defined inductively: for types denoted by $x$ and $y$, we define the
type $z\equiv x\to y$ as the
type of maps transforming operations of type $x$ to operations of type $y$. For example,
$A\to B$ specifies the type of quantum channels from system $A$ to system $B$, $(A_1\to
A_2)\to (A_3\to A_4)$ is the type of superchannels, transforming channels $A_1\to A_2$
into channels $A_3\to A_4$, etc. Each type is therefore composed over a sequence of
elementary types $A_1,\dots, A_n$ which will be assumed all different (but they might be
isomorphic). Trivial elementary types will be denoted by $I$ (that is,
$\mathcal H_{I}\equiv \mathbb C$). 

Using the Choi representation, the set of all maps  of a given type $x$ can be
identified with a subset of positive operators on the tensor product of all involved
elementary quantum systems. It was proved in
\cite{perinotti2017causal,bisio2019theoretical} that this set is determined by a trace
normalization condition and linear constraints. More precisely, let $T_1(x)$ be the set of the Choi operators of 
(deterministic) maps  of type $x$, constructed over elementary types $A_1,\dots, A_n$, and
let $\mathcal H_x:=\otimes_i \mathcal H_{A_i}$. Then we have \cite[Prop.~1]{bisio2019theoretical}
\begin{equation}\label{eq:BPtypes}
R\in T_1(x) \iff R\ge 0,\ R=\lambda_xI_{\mathcal H_x}+X_x,\qquad X_x \subseteq \Delta_x,
\end{equation}
where $\lambda_x>0$ is given by the construction of the type $x$, and $\Delta_x$ is a
certain subspace of traceless hermitian operators in $B(\mathcal H_x)$. This subspace is
determined as follows:  For each of the elementary systems $A_i$, the space $B_h(\mathcal
H_{A_i})$ of hermitian operators on  $\mathcal H_{A_i}$ has an orthogonal decomposition
(with respect to the Hilbert-Schmidt inner product) into multiples  of identity
and the  traceless part\footnote{Note that in
\cite{perinotti2017causal,bisio2019theoretical,apadula2024nosignalling}, the indices 0 and 1
in the definition of the spaces are interchanged. Unfortunately, these indices were flipped in
\cite{jencova2024onthestructure}. We keep this flipped notation for consistency.}:
\[
B_h(\mathcal H_{A_i})=L_{i,0}\oplus L_{i,1},\qquad L_0=\mathbb RI,\ L_1=\{X\colon
\mathrm{Tr}(X)=0\}.
\]
We therefore have an orthogonal decomposition of $B_h(\mathcal H_x)$:
\[
B_h(\mathcal H_x)=\bigotimes_{i=1}^n B_h(\mathcal H_{A_i})=\bigoplus_{s\in\{0,1\}^n} L_s,\qquad L_s:=L_{1,s_1}\otimes \dots\otimes
L_{n,s_n}.
\]
By \cite[Cor.~3]{bisio2019theoretical}, there is some subset of strings $D_x\subseteq \{0,1\}^n$ such
that $\Delta_x$ is the direct sum of subspaces  
\begin{equation}\label{eq:BPcombinatorial}
\Delta_x=\bigoplus_{s\in D_x} L_s.
\end{equation}

\subsection{Projective characterization}

The characterization of deterministic maps in \eqref{eq:BPtypes} and
\eqref{eq:BPcombinatorial} can be written in the form
\begin{equation}\label{eq:BPsubspace}
R\in T_1(x) \iff R\in S_x\cap B(\mathcal H_x)^+,\ \Tr[R]=c_x:=\dim(\mathcal H_x)\lambda_x, 
\end{equation}
where $S_x:=\Delta_x\oplus \mathbb R I_{\mathcal H_x}=\bigoplus_{s\in D_x\cup\{\theta\}}
L_s$, here $\theta=\theta_n$ is the string of zeros (see Appendix \ref{app:fn} for more
notations and definitions). Therefore, the type $x$ can be characterized by the linear subspace $S_x$ and a 
positive constant $c_x$. Replacing the subspace $S_x$ by the orthogonal projection $P_x$
onto it, we arrive at the basic idea of the projective characterization of higher order maps in
\cite{hoffreumon2026projective, milz2024characterising}. The projection $P_x$ can be rather
straightforwardly expressed as the sum of orthogonal projections onto $L_s$:
\begin{equation}\label{eq:BPprojection}
P_x=\sum_{s\in D_x\cup\{\theta\}} P_s,\qquad P_s=P_{1,s_1}\otimes \dots \otimes P_{n,s_n},
\quad P_{i,0}(X)=\frac{\Tr[X]}{\dim(\mathcal H_{A_i})} I_{\mathcal H_{A_i}},\
P_{i,1}=P_{i,0}^\perp.
\end{equation}

A general framework describing completely positive
transformations that map between sets of positive operators given by affine restrictions
 was introduced in \cite{milz2024characterising}. For the special case of higher order
 types, the corresponding
projections are expressed in terms of linear combinations (with coefficients in
$\{0,1,-1\}$) of trace-and-replace projections that (up to a
permutation of the spaces) can be written in the form
\begin{equation}\label{eq:MQprojection}
\Pi_T:= \bigotimes_{i\in T}P_{i,0}\bigotimes_{j\in T^C} id_{B_h(\mathcal H_{A_j})},\qquad T\subseteq [n].
\end{equation}
See \cite[Examples 1-4, 6-8]{milz2024characterising}. In general, this can be seen by iterating the general
from of such projections in \cite[Thm.~2]{milz2024characterising}, but also from \eqref{eq:BPprojection}.

It was  observed in \cite{hoffreumon2026projective} that all the projections corresponding
to higher order types form a boolean algebra, with a natural interpretation of the
corresponding algebraic operations and order structure. Moreover, further operations on these algebras are introduced, representing no-signalling, full signalling and  one-way signalling compositions of types, resulting in
a structure of a BV-logic, see also \cite{simmons2022higher}. This approach can
be extended to any families of mutually commuting projections representing basic systems.
In \cite{jencova2024onthestructure} it was observed that the algebraic structure of the higher order types can be 
obtained from the combinatorial representation of types in \cite{bisio2019theoretical}.
The corresponding constructions are explained in the next section.

\section{Type functions}\label{sec:functions}

In the framework of higher order types, there are two more operations defined using the trivial type
$I$: the dual type $\bar x$ and the tensor product of types $x\otimes y$  defined as
 \cite{bisio2019theoretical}
\[
\bar x\equiv x\to I,\qquad x\otimes y\equiv \overline{x\to \bar y}.
\]
The dual defines the type of \emph{effects} over $x$, describing transformations that map higher
order maps of type $x$ to probabilities (or to 1, in the deterministic case). The tensor product corresponds to
\emph{no-signalling composition} of two maps of types $x$ resp.  $y$.

We then have $\bar{\bar x}=x$ and $x\to y= \overline{x\otimes \bar y}$, it is therefore
clear that we can use $\otimes$ and $\bar{\cdot}$ instead of $\to$ in the construction of
types over a set of elementary types. In \cite{jencova2024onthestructure}, the higher
order types were constructed in the setting of the category of affine subspaces, where the
tensor product and the dual come from a *-autonomous structure on the category. We will
adopt this approach in the more restricted quantum setting here, but note that exactly the
same constructions apply for higher order maps of the classical theory, see
\cite[Sec.~2.2.1]{jencova2024onthestructure}.

 To avoid redundancies, we will assume that all the involved elementary types are nontrivial.
Any type is then uniquely given as a linear term whose variables are the elementary types $A_1,\dots,
A_n$ and the operations are $\otimes$ and $\bar{\cdot}$, e.g. 
$x=\overline{(A_{i_1}\otimes \bar{A}_{i_2})}\otimes A_{i_3}\otimes
(\overline{A_{i_4}\otimes A_{i_5}}\dots)$.
Notice that the  constant $c_x$ in \eqref{eq:BPsubspace} is also determined from this term as
$c_x:=\Pi_{i\in I_x}\dim(\mathcal H_{A_i})$, where
$I_x\subseteq [n]$ is the  set of indices of the elementary types that are subject to
taking the dual an odd number of times. Such indices will be called \emph{inputs} of $x$, the
indices in $O_x:=[n]\setminus I_x$ will be called \emph{outputs} of $x$. The intuition
behind this terminology is that any higher order map can be seen as a channel from the
input system $\otimes_{i\in I_x}\mathcal H_{A_i}$ to  the output $\otimes_{j\in O_x}
\mathcal H_{A_j}$.

In \cite{jencova2024onthestructure}, the types were related to certain boolean functions. We will use the notation
\begin{equation}\label{eq:Fen}
\Fe_n:=\{f:\{0,1\}^n\to \{0,1\}\colon f(\theta)=1\}.
\end{equation}
Some operations on $\Fe_n$ and their properties  are collected in Appendix \ref{app:fn}. In
particular, $\Fe_n$ is a boolean algebra with a bottom element denoted as $p_n$, top element
$1_n$ (the string of 1's), operations $\vee$ and $\wedge$ defined as pointwise maximum and minimum, and with a
complement given as $f^*=1-f+p_n$. Moreover, a tensor product $\otimes :\Fe_m\times \Fe_n\to \Fe_{m+n}$ is
defined in a natural way. 

Each type $x$ can be related to a unique element of $\Fe_n$ as follows.  Consider
the subset of strings $D_x$ corresponding to the type $x$ as in \eqref{eq:BPcombinatorial}
and let $f$ be  the characteristic function of $D_x\cup\{\theta\}$, that is,
$f:\{0,1\}^n\to \{0,1\}$ and 
\[
f(s)=\chi_{D_x\cup\{\theta\}}(s):=\begin{dcases} 1 & \text{if } s\in D_x\cup\{\theta\}\\
 0 &\text{otherwise}.
\end{dcases}
\]
Then it is clear that $f\in \Fe_n$ and we have 
\begin{equation}\label{eq:function_subspace}
S_x= S_f\equiv \bigoplus_{s\in \{0,1\}^n} f(s)L_s.
\end{equation}
Note that the function $f$ only contains information about the \emph{type structure} of $x$
(see \cite{bisio2019theoretical}). The  particular elementary types $A_1,\dots,A_n$ that are
contained in $x$ are an additional information needed for the construction of the subspace
$S_f$, and we should actually write $S_f=S_f(A_1,\dots,A_n)$ and $L_s=L_s(A_1,\dots,A_n)$. 
For simplicity, we will suppress this notation.

So assume that we are given a collection of elementary types $A_1,\dots, A_n$, with
corresponding Hilbert spaces $\mathcal H_i:=\mathcal H_{A_i}$ and $\mathcal H:=\otimes_i
\mathcal H_i$. 
By orthogonality of the subspaces $L_s$, it is clear that the map $f\mapsto S_f$ is
injective. It is also clear that $S_f$ is a linear subspace in $B_h(\mathcal H)$ for
any $f\in \Fe_n$, and this subspace always contains a multiple of the identity (since
$f(\theta)=1$). Moreover, we have
\begin{equation}\label{eq:Sf}
S_{f\wedge g}=S_f\wedge S_g,\ S_{f\vee g}=S_f\vee S_g,\ S_{f^*}=S_f^\perp\vee \mathbb R
I_{\mathcal H},\ S_{f\otimes
g}=S_f\otimes S_g,
\end{equation}
here $\wedge$ and $\vee$ on the right hand sides of the equalities are the usual lattice
operations on linear subspaces,  and
$\otimes$ is the tensor product.
Replacing the subspace $S_f$ by the corresponding projection $P_f=\sum_s
f(s)P_s$, we obtain a collection of mutually commuting orthogonal projections
$\{P_f\}_{f\in \Fe_n}$, such that $f\mapsto P_f$ is a
representation of the boolean algebra $\Fe_n$ in the lattice of projections on
$B_h(\mathcal H)$, and this representation respects the tensor product structure on both
sides. 

In \cite{jencova2024onthestructure}, functions $f\in \Fe_n$ such that $S_f$ is a subspace corresponding to a higher
order type were called type functions. Using \eqref{eq:BPcombinatorial} and
\eqref{eq:BPsubspace}, we obtain the following definition.

\begin{defi}\label{defi:type_functions}
A function $f\in \Fe_n$ is a \emph{type function} if $\supp(f)=D_x\cup\{\theta\}$ for some
type $x$, here $\supp(f):=\{s\in \{0,1\}^n\colon f(s)=1\}$. 
 The set of all type functions in $\Fe_n$ is denoted as $\Te_n$.

\end{defi}

As noted in \cite{jencova2024onthestructure}, we have  $\Te_1=\Fe_1=\{1_1,p_1\}$ but
$\Te_n\subsetneq \Fe_n$ for $n\ge 2$, so that not all $f$ correspond to higher order
types

The type function related to a type $x$ can be constructed from its expression as a term over its elementary 
types $A_1,\dots, A_n$ as follows. Note that for an elementary type $A$ and
arbitrary types $x$, $y$ with corresponding type functions $f$ and $g$, we have
(\cite[Prop.~3.2]{jencova2024onthestructure}) 
\begin{equation}\label{eq:terms}
A\equiv 1_1,\qquad \bar x\equiv f^*=1-f+p_n,\qquad (x\otimes y)\equiv (f\otimes g).
\end{equation}
The related type function of a term $x$ is then obtained by applying the rules in \eqref{eq:terms}.
We can see from this construction that any type function $f\in 
\Te_n$ is, up to a  permutation, either a product of two type functions with smaller $n$, or a complement of such a product. 
This fact is frequently used in inductive proofs of properties of type functions.

Let $f\in \Te_n$ and let $x$ be the corresponding type. We will denote $I_f:=I_x$,
$O_f:=O_x$ the
corresponding sets of input and output indices. These sets can be  easily obtained from
$f$. Indeed, let us 
define the strings $e^i\in \{0,1\}^n$, $i\in [n]$,  given as 
\begin{equation}\label{eq:ei}
e^i_j=\delta_{i,j},  \qquad i,j=1,\dots n. 
\end{equation}
If we need to stress the value of $n$, we use the notation $e^{i:n}$.

\begin{prop}\label{prop:inout}\cite[]{jencova2024onthestructure} Let $f\in \Te_n$ and let $i\in [n]$. Then 
 $i\in I_f$ if and only if $f(e^i)=0$. 
\end{prop}

This result  can be straightforwardly extended to define a decomposition of $[n]$ into the
'input/output indices' relative to  any $f\in \Fe_n$:
\begin{equation}\label{eq:inout}
I_f:=\{i\in [n]\colon f(e^i)=0\},\qquad O_f:=\{j\in [n]\colon f(e^j)=1\}.
\end{equation}
For $f,g\in \Fe_n$, we clearly have $I_f=I_g$ if and only if $O_f=O_g$, and in this case,
$I_{f\wedge g}=I_{f\vee g}=I_f$.

Given any $f\in \Fe_n$ and elementary types $A_1,\dots,A_n$, we may now define 
\[
T_1(f):= \{C\in S_f\cap B(\mathcal H)^+\colon \Tr[C]=c_f:=\Pi_{i\in I_f} \dim(\mathcal H_i)\},
\]
which can always be interpreted as  convex sets of Choi operators of some completely
positive maps $B(\mathcal H_{I_f})\to B(\mathcal H_{O_f})$. These maps are not necessarily
trace preserving, but such sets always contain the constant channel sending all states to
the maximally mixed state in $B(\mathcal H_{O_f})$ (corresponding to the element
$c_f\dim(\mathcal H)^{-1}I_{\mathcal H}\in T_1(f)$).

From \eqref{eq:BPsubspace} and \eqref{eq:function_subspace}, we see that if $f$ is the
type function of a type $x$, we have $T_1(f)=T_1(x)$.

\begin{lemma}\label{lemma:Tf} Let $f,g\in \Fe_n$, $O_f=O_g$. Then
\begin{align*}
T_1(f\wedge g)&=T_1(f)\cap T_1(g)\\
T_1(f\vee g)&=\{C\in B(\mathcal H)^+\colon
C=uC_1+(1-u)C_2,\ C_1\in T_1(f),\ C_2\in T_1(g),\ u\in \mathbb R\}.
\end{align*}
In particular, $T_1(f)\subseteq  T_1(g)$  if and only if $f\le g$. 

\end{lemma}

Notice that the set $T_1(f\vee g)$ consists of positive operators that are \emph{affine}
combinations of Choi matrices of types $f$ and $g$, which is not necessarily expressible
as a convex combination. An important example are some instances of \emph{causally nonseparable} processes
that are expressible as an affine combination, but not as a convex combination,  of causally ordered processes
\cite{araujo2015witnessing,oreshkov2016causal}.

\begin{proof} By \eqref{eq:Sf}, we have $S_{f\wedge g}=S_f \wedge S_g=S_f\cap S_g$ and
$S_{f\vee g}=S_f\vee S_g$.
The first equality is now immediate from the definition and $I_f=I_g$. For the
second equality, it is easy to see that all elements of the form on the right are
contained in $T_1(f\vee g)$. For the converse, note that  any $C\in T_1(f\vee g)$ is a positive operator of the form
$C=X_1+X_2$, where $X_1\in S_f$ and $X_2\in S_g$. Let $\mu>0$ be such that $-\mu
I_{\mathcal H}\le X_1,X_2\le \mu I_{\mathcal H}$, then $\tilde X_1:=X_1+\mu I_{\mathcal
H}\in S_f\cap B(\mathcal H)^+$, $\tilde X_2:=\mu I_{\mathcal H}-X_2\in S_g\cap B(\mathcal
H)^+$ and $C=\tilde X_1-\tilde X_2$. By normalization, we see that there are some $u,v>0$ 
and $C_1\in T_1(f)$, $C_2\in T_1(g)$ such that $C=uC_1-vC_2$. Using the trace condition,
 we obtain that we must have $u-v=1$, so that $-v=1-u$. 

For the  last assertion, we first observe that $T_1(f)$ generates $S_f$ (as also seen in the
argument above), and since the map $f\mapsto S_f$ is injective, we have $T_1(f)=T_1(g)$ if and only if $f=g$. The statement  is now straightforward from the first equality.

\end{proof}

\subsection{The M\"obius transform} \label{sec:mobius}

An important example of a type function is the following, see \cite[Ex.~4.1]{jencova2024onthestructure}. For any subset  $T\subseteq [n]$, put
\[
p_T(s)=\Pi_{i\in T}(1-s_i),\qquad s\in \{0,1\}^n.
\]
It is clear that $p_T\in \Fe_n$, and the input/output indices of $p_T$ are $I=T$, $O=T^C$. 
This function corresponds to the type of constant channels: 
given elementary types $A_1,\dots, A_n$, the deterministic maps of this type are those that
map any state $\rho$ on $\mathcal H_T:=\otimes_{i\in T}\mathcal H_i$ to some fixed state 
on $\mathcal H_{T^C}$. In particular, for $T=\emptyset$, we obtain $p_\emptyset=1_n$,
corresponding to state preparation channels, and for $T=[n]$ we get $p_{[n]}=p_n$,
describing the trace. For any $T\subseteq [n]$, the corresponding projection onto the
subspace $S_{p_T}$ is 
$P_{p_T}=\Pi_T$ given in \eqref{eq:MQprojection}. For the  complement $p_T^*$, the set 
$T_1(p_T^*)$ consists of 
all channels $B(\mathcal H_{T^C})\to B(\mathcal H_T)$. 

With some abuse of notations,  we
will use the notation $I\to O$ for the type structure of channels where the input/output
spaces comprise of elementary types indexed by $i\in I$/$j\in O$, so that the
corresponding type function is $p_O^*$.  For a given set of elementary
types $A_1,\dots, A_n$, we obtain the type  $\otimes_{i\in I} A_i \to
\otimes_{j\in O}A_j$.  Similar notations will
be used for other type structures of higher order maps.

Any function $f\in \Fe_n$ (in fact, any function $\{0,1\}^n\to \mathbb R$) can be uniquely 
expressed as a (real) linear combination of the functions $p_T$, $T\subseteq [n]$. 
The M\"obius transform of a  function $f$ can be defined as the
unique assignment $[n]\supseteq T\mapsto \hat f_T$ of the corresponding coefficients, that
is, 
\begin{equation}\label{eq:mobius}
\hat f_T:= \sum_{\substack{s\in \{0,1\}^n\\
s_j=1 \forall j\in T^C}}(-1)^{\sum_{j\in T}s_j}f(s),\qquad    f=\sum_{T\subseteq [n]} \hat f_T p_T.
\end{equation}
The properties of $\hat f$ for type functions were explored in
\cite{jencova2024onthestructure} and will be summarized in sections below. In particular,
it was proved that we always have $\hat f_T\in \{-1,0,1\}$ for a type function $f$. Given
the uniqueness of the expression of $f$ in terms of $p_T$, and the fact that
$P_{p_T}=\Pi_T$, we arrive at the following observation.

\begin{obs} The M\"obius transform provides a transition between the Bisio-Perinotti
\cite{bisio2019theoretical} and the Milz-Quintino \cite{milz2024characterising} characterizations of
quantum higher order types.

\end{obs}

As shown in \cite[]{jencova2024onthestructure}, the M\"obius transform helps us to characterize an important higher order type. 
A type function is  called a \emph{chain type} if there is a chain of subsets
$S_0\subsetneq \dots\subsetneq S_N\subseteq [n]$ with $N$ even, such that 
\begin{equation}\label{eq:chain}
f=\sum_{i=0}^N (-1)^i p_{S_i}.
\end{equation}
The corresponding type is a quantum comb (\cite{chiribella2009theoretical}), with the successive input and output spaces
given as in the following diagram

\begin{center}
\includegraphics{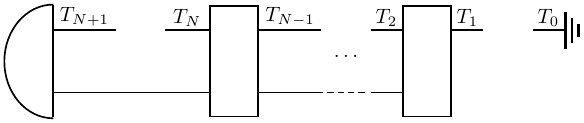}
\end{center}
where  $T_0:=S_0$, $T_i:=S_i\setminus S_{i-1}$, $i=1,\dots,N$,
$T_{N+1}=[n]\setminus S_N$. Note that the sets $T_0$ and $T_{N+1}$ may be empty, the
corresponding spaces in the diagram are then  trivial. As can be also seen from the
diagram, the input and output indices in this case are given as 
\[
I_f=\bigcup_{k=0}^{\frac N2} T_{2k},\qquad  O_f=\bigcup_{k=0}^{\frac N2} T_{2k+1}.
\]
In the case $n\le 3$ all type functions are chain types, so that all higher order types over at
most three elementary types are combs.

\subsection{The causal product} \label{sec:causal}

Let $f_i\in \Fe_{n_i}$, $i=1,2$, $n=n_1+n_2$  and let us fix the decomposition
$[n]=[n_1]\oplus[n_2]$, with the corresponding decomposition of strings
$s=s^{1}s^{2}$ (see Appendix \ref{app:fn}). We define the causal product as
\begin{equation}\label{eq:causalp}
f_1\vtl f_2:=(f_1-p_{n_1})\otimes 1_{n_2}+p_{n_1}\otimes f_2 
\end{equation}
using the tensor product of boolean functions. On strings, the resulting function acts as 
\begin{equation}\label{eq:causalp_s}
(f_1\vtl f_2)(s)=\begin{dcases} f_1(s^1) &\text{if } s^1\ne \theta\\
f_2(s^2) & \text{otherwise}.
\end{dcases}
\end{equation}
In particular, this implies that $f_1\vtl f_2\in \Fe_{n_1+n_2}$. We similarly  define 
\[
f_2\vtl f_1:= f_1\otimes p_{n_2}+1_{n_1}\otimes (f_2-p_{n_2})
\]
with a formula similar to \eqref{eq:causalp_s}. Note that this definition is with respect
to the fixed decomposition $[n]=[n_1]\oplus[n_2]$ and concatenation of strings $s=s^1s^2$,
and the function $f_2$ always acts on the second part $s^2$.

As explained in
\cite{jencova2024onthestructure}, the causal product $f_1\vtl f_2$ of type functions   can be interpreted as a composition
of types in such a way that that the map corresponding to $f_1$ is in the causal future of
the map corresponding to $f_2$ (similarly for $f_2\vtl f_1$). This definition was inspired by the 'prec' connector
$\preceq$ in
the  projective characterization of \cite{hoffreumon2026projective}, and we have
\[
P_{f_1\vtl f_2}=P_{f_1}\succeq P_{f_2},\qquad P_{f_2\vtl f_1}=P_{f_1}\preceq P_{f_2}.
\]
Let us summarize some properties of the causal product that will be needed in the sequel.
\begin{lemma}\label{lemma:causalp}\cite[Lemmas 4.11 and 4.12]{jencova2024onthestructure} 
\begin{itemize}
\item[(i)] $(f_1\vtl f_2)^*=f_1^*\vtl f_2^*$,
\item[(ii)] $(f_1\vtl f_2)\vtl f_3=f_1\vtl(f_2\vtl f_3)=:f_1\vtl f_2\vtl f_3$,
\item[(iii)] $f_1\otimes f_2=(f_1\vtl f_2)\wedge (f_2\vtl f_1)$,
$f_1\parr f_2=(f_1\vtl f_2)\vee (f_2\vtl f_1)$. 
\item[(iv)] $(f_1\vee f_2)\vtl (f_3\vee f_4)=(f_1\vtl f_3)\vee (f_2\vtl f_4)=(f_1\vtl
f_4)\vee (f_2\vtl f_3)$. A similar statement holds for $\wedge$. 
\end{itemize}
\end{lemma}

We will also use the following extension of Lemma \ref{lemma:causalp} (iii). 

\begin{lemma}\label{lemma:causalp_ext} Let $f_1,g_1\in \Fe_{n_1}$, $f_2,g_2\in \Fe_{n_2}$
and assume  that $f_1\le g_1$ and $f_2\le g_2$. Let us fix the decomposition
$[n]=[n_1]\oplus [n_2]$. Then we have
\[
f_1\otimes f_2= (f_1\vtl g_2)\wedge (f_2\vtl g_1).
\]

\end{lemma}

\begin{proof} Using Lemma \ref{lemma:causalp} (iv), we have 
\[
f_1\vtl f_2=f_1\vtl (f_2\wedge g_2)=(f_1\vtl f_2)\wedge (f_1\vtl g_2)\le (f_1\vtl g_2)
\]
and similarly $f_2\vtl f_1\le f_2\vtl g_1$. Using part (iii) of the lemma, we get
\[
f_1\otimes f_2=(f_1\vtl f_2)\wedge (f_2\vtl f_1)\le (f_1\vtl g_2)\wedge (f_2\vtl g_1).
\]
For the opposite inequality, it is enough to show that $(f_1\otimes f_2)(s)=0$ implies
that $(f_1\vtl g_2)(s)\wedge (f_2\vtl g_1)(s)=0$. So let  $s=s^1s^2$ be a string such that
$(f_1\otimes f_2)(s)=0$. If  $f_1(s^1)=0$, then $s^1\ne \theta$, so that $f_1\vtl
g_2(s)=f_1(s^1)=0$. In the case that $f_1(s^1)=1$ we must have $f_2(s^2)=0$, which similarly
implies that $f_2\vtl g_1(s)=0$. This shows that $(f_1\vtl g_2)(s)\wedge (f_2\vtl
g_1)(s)=0$.

Note that we did not need any assumptions on $g_1$ and $g_2$ for this inequality.

\end{proof}

The next result shows that the causal product of type functions is not necessarily a type
function.

\begin{prop}\label{prop:causal_product_type} Let $f\in \Te_n$, $g\in \Te_m$.
Then 
$f\vtl g$ is a type function if and only if $f$ or $g$ is a chain type.  If both $f$ and
$g$ are chain types, then $f\vtl g$ is a chain type as well.

\end{prop}

\begin{proof}  The 'if' part, and the last statement were proved in
\cite[Prop.~4.15]{jencova2024onthestructure}. The
'only if' part is proved in Appendix \ref{app:structure}, using the structure posets of a
type introduced in the next section.

\end{proof}

\subsection{The structure  posets related to a type}\label{sec:structure}

In \cite{jencova2024onthestructure}, we have used the M\"obius transform of a type
function $f$  to obtain certain posets with labelled vertices. These posets uniquely
represent the type and show some inherent orderings of the involved elementary systems.  

\begin{defi}\label{defi:structure poset}
Let $f\in \Te_n$ and let $T\mapsto \hat f_T$ be the M\"obius transform. The
\emph{(full) structure 
poset} $\Pe_f$ of $f$ is the subposet in the lattice $2^n$ of all subsets of $[n]$, consisting of
all $T\subseteq [n]$ such that $\hat f_T\ne 0$. For $T\in \Pe_f$, the set of \emph{labels}
of $T$ is defined as 
\[
L_T:=T\setminus \cup \{S\in\Pe_f \colon S\subsetneq T\},
\]
(that is, the labels of $T$ are all indices in $T$ that are not contained in any 
element of $\Pe_f$ that is strictly below $T$). If needed, we will use the superscript
$L^f_T$ to indicate that the label set is connected to $f$. 
The \emph{reduced structure  poset} of $f$ is the subposet $\Pe_f^0\subseteq \Pe_f$, consisting of
$\emptyset$ (if present in $\Pe_f$) and  
all $T\in \Pe_f$ such  that $L_T\ne\emptyset$.

\end{defi}

The posets are  represented by their Hasse diagrams, which are defined as graphs whose
vertices are elements of the respective poset and  there is an edge between vertices $p$ and
$q$ if $p$ is covered by $q$ (that is, $p\le q$ and if $p\le r\le q$, then $p=r$ or
$q=r$). Moreover,  if $p\lneq r$, then $r$ is drawn above
$p$.  Two posets are isomorphic if and only if they have the same Hasse diagrams. We will
also label the vertices with their respective labels.  Some examples are presented further below.

We collect some  general properties of the structure posets  that were proved in
\cite{jencova2024onthestructure}.
First of all, the structure poset $\Pe_f$
is a graded poset of even rank, which is denoted by $r(f)$. This means that every maximal
chain (i.e. a maximal totally
ordered subposet) in $\Pe_f$ has the same length equal to $r(f)$. Moreover, there  is a rank function
$\rho:\Pe_f\to [0,r(f)]$, assigning to each element its rank. The rank function is
strictly  increasing along any chain, with minimal elements having rank 0. Moreover, for the type
function $f$, we have 
\begin{equation}\label{eq:structure_poset}
f=\sum_{T\in \Pe_f} (-1)^{\rho(T)}p_T.
\end{equation}
This equality shows that the type
function $f$ is fully determined by its structure poset. Furthermore, by uniqueness of the
M\"obius transform, this shows that $\hat f_T$ has values in $\{-1,0,1\}$.

It was shown in
\cite[Thm.~4.22]{jencova2024onthestructure} that any type function is fully determined  by its
reduced poset as well, but in general this relation is not as straightforward as 
\eqref{eq:structure_poset}. In particular, the function $f$ is a chain type  iff $\Pe_f=\Pe_f^0$ iff 
$\Pe_f$ is totally ordered  iff maps of type $f$ are causally ordered (that is,  combs), see also
\eqref{eq:chain} and the paragraphs around it.

We now define the rank of an index with respect to a type function.
For $i\in [n]$, let us denote
\[
\mathbb T_i^f:=\{T\in \Pe_f\colon i\in L_T\}.
\]
All elements $T\in \mathbb T^f_i$ have the same rank, which is called the \emph{rank
 of
the index $i$ in $f$} and is denoted by $r_f(i)$. If $\mathbb T_i^f=\emptyset$, we set $r_f(i)=r(f)+1$.
The index $i\in [n]$ is an input if and only if $r_f(i)$ is even. Indices in $\cap\Pe_f$
(which are precisely those indices that label all the minimal elements in $\Pe_f$) are called \emph{free inputs} of $f$. The indices not contained in any label have odd rank
$r(f)+1$, these are called the \emph{free outputs}. The set of free inputs/outputs of $f$ will be
denoted by $I^F_f/O^F_f$. The free inputs and outputs together are called the \emph{free
indices} of $f$. \cite[Ex.~4.4]{jencova2024onthestructure} shows how the free indices are
manifesting in the higher order maps of type $f$.

\begin{exm}[Simplest types]\label{exm:hasse} The simplest examples in $\Te_n$ are states
corresponding to the state space in $B(\mathcal H)$, with type 
function $f\equiv 1$ and $\Pe_f=\Pe_f^0=\{\emptyset\}$, and the trace, with $f=p_n$ and
$\Pe_f=\Pe_f^0=\{[n]\}$. For the  type of constant channels with inputs in $I\subseteq
[n]$, we have $f=p_I$ and the corresponding posets have again a single element
$\Pe_f=\Pe_f^0=\{I\}$, note that this type has only free outputs, so that $O_f^F=O=I^C$. The labelled Hasse
diagrams are equally simple:

\begin{center}
\vcbox{\includegraphics[scale=0.9]{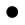}}
\vcbox{$\qquad$}
\vcbox{\includegraphics[scale=0.9]{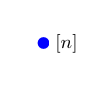}}
\vcbox{$\qquad$}
\vcbox{\includegraphics[scale=0.9]{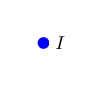}}
\end{center}
\end{exm}
Here, as in all the diagrams below, we use blue color for the input vertices, that is, labelled
vertices with even rank, and red color for the output vertices (labelled vertices with odd
rank). The vertex corresponding to $\emptyset$ is left black. Vertices that do not belong to
the reduced poset are gray, together with the adjacent edges not belonging to $\Pe_f^0$. 

\begin{exm}[Channels and combs]\label{exm:combs}  As we have seen,  channels $I\to O$ have type
function $f=p_O^*$, and $\Pe_f=\Pe_f^0=\{\emptyset\subsetneq O\subsetneq [n]\}$. 
We also give below the diagram of a  2-comb for $n=4$,  with the corresponding poset $\Pe_f=\Pe_f^0=\{\emptyset\subsetneq \{1\}\subsetneq
\{1,2\}\subsetneq \{1,2,3\}\subsetneq [4]\}$. The related higher order maps  $(3\to
2)\to (4\to 1)$ are called quantum supermaps (or superchannels).
\begin{center}
\vcbox{\includegraphics[scale=0.9]{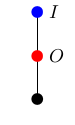}}
\vcbox{$\qquad$}
\vcbox{\includegraphics{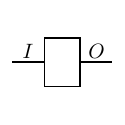}}
\vcbox{$\qquad,$}
\vcbox{$\qquad$}
\vcbox{\includegraphics[scale=0.9]{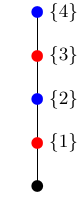}}
\vcbox{$\qquad$}
\vcbox{\includegraphics{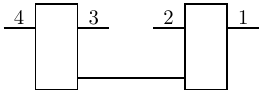}}
\end{center}

\end{exm}

\begin{exm}[Nonsignalling channels and process matrices]\label{exm:nspm} As mentioned
before, all type functions for $n\le 3$ are chain types. We next show the
structure posets of type functions in $\Te_{4,\{1,3\}}$ that are not chain types.
Let $\gamma_2=1_2-p_{\{1\}}+p_{\{1,2\}} \in \Te_2$ be the type function of channels $2\to 1$, then 
the tensor product $f_{ns}:=\gamma_2\otimes
\gamma_2\in \Te_4$,  describes nonsignalling channels $\{2,4\}\to \{1,3\}$. 
Let $\tilde\gamma_2:=1_2-p_{\{2\}}+p_{\{1,2\}}\in \Te_2$, corresponding to channels $1\to 2$,
then $f_{pm}:=(\tilde \gamma_2\otimes \tilde \gamma_2)^*$ describes the process matrices, which
are positive functionals mapping all nonsignalling channels to 1. The diagrams below
depict these two types, together with the diagram of the type of full signalling channels
$p_{\{1,3\}}^*$:
\begin{center}
\vcbox{\includegraphics[scale=0.9]{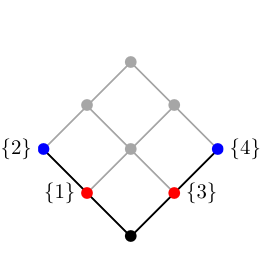}}
\vcbox{$\qquad$}
\vcbox{\includegraphics[scale=0.9]{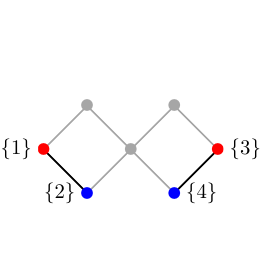}}
\vcbox{$\qquad$}
\vcbox{\includegraphics[scale=0.9]{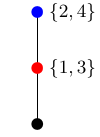}}
\end{center}

\end{exm}

\begin{exm}[Process matrices with global past and future]\label{exm:pm_pf} Let
$f_{pm}=(\tilde \gamma_2\otimes\tilde \gamma_2)^*\in \Te_4$ be the type function of process matrices as in
the previous example. Then $1_1\vtl f_{pm}\vtl p_1\in \Te_6$ is the type function for process
matrices with global past and future, with the corresponding diagram:
\begin{center}
\vcbox{\includegraphics[scale=0.9]{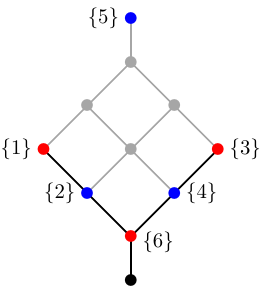}}
\end{center}
Here we have applied a permutation on the indices to keep the labels of the process matrix
$f_{pm}$ as in the previous example. The global past has label 5, the system of the global
future has label 6. 
\end{exm}

\medskip

Under the basic operations on types, the structure posets behave as follows: 
\[
\Pe_{f\circ \sigma}\simeq \Pe_f,\qquad \Pe_{f*}=\Pe_f \triangle \{\emptyset, [n]\},\qquad \Pe_{f\otimes g}=\Pe_f\times \Pe_g,
\]
here $\sigma\in \permut_n$, $\triangle$ denotes symmetric difference of sets and $\times$ is the direct product of
posets. For the causal product, $\Pe_{f\vtl g}$ is related to the ordinal sum $\Pe_f\star
\Pe_g$, with some subtleties according to whether $[n]\in \Pe_f$ and $\emptyset \in
\Pe_g$, see \cite[Prop.~4.14]{jencova2024onthestructure} for details. Definitions of the
above poset compositions can be also found in \cite[App.~A.2]{jencova2024onthestructure}. The results for reduced posets
are summarized in Appendix \ref{app:structure}, together with some further properties of
the structure posets.  These properties will be applied to prove Proposition
\ref{prop:causal_product_type}.

\section{Subtypes}\label{sec:subtypes}

We have defined for each $f\in \Fe_n$ a corresponding set $T_1(f)$  of (Choi operators of) completely positive
maps. We now consider functions for which all these maps are channels $B(\mathcal H_{I_f})\to
B(\mathcal H_{O_f})$. As we
have seen at the beginning of Section \ref{sec:mobius}, this means that 
$T_1(f)\subseteq T_1(p_{O_f}^*)$, which by  Lemma \ref{lemma:Tf} is equivalent to $f\le
p_{O_f}^*$. Requiring the same for $f^*$, noticing that
$I_{f^*}=O_f$, $O_{f^*}=I_f$,  leads us to the following
definition.

\begin{defi}\label{defi:subtype} A function $f\in \Fe_n$ is called a
\emph{subtype} if 
\[
p_{I_f}\le f\le p_{O_f}^*.
\]

\end{defi}
By  \cite[Lemma 4.1]{jencova2024onthestructure},  any type function is a subtype.
Note also that for any decomposition $[n]=I\sqcup O$ we have $p_I\le p_O^*$ and any
function $f$ in
the interval between them must satisfy $I=I_f$ and $O=O_f$, so that any such function is a
subtype. 
 
 \begin{prop}\label{prop:subtype_dual_tensor} Let $f\in \Fe_n$, $g\in
 \Fe_m$ be subtypes.
 Then
 \begin{enumerate}
\item[(i)] $f^*$ is a subtype, with $I_{f^*}=O_f$, $O_{f^*}=I_f$.
\item[(ii)] For any $\sigma\in \permut_n$, $f\circ\sigma$ is a subtype, with $I_{f\circ
\sigma}=\sigma^{-1}(I_f)$, $O_{f\circ\sigma}=\sigma^{-1}(O_f)$.
\item[(iii)] $f\otimes g$, $f\parr g$, $f\vtl g$ and $g\vtl f$  are  subtypes, with inputs in $I=I_f\oplus
I_g$ and outputs in $O=O_f\oplus O_g$.
\item[(iv)] The set of all subtypes with the same inputs/outputs is a distributive
lattice.
 \end{enumerate}

 \end{prop}

\begin{proof} (i) is immediate by the properties of the complement. For (ii), notice that
for any index $i$,
$\sigma(e^{i})=e^{\sigma(i)}$ and hence $f\circ \sigma(e^i)= f(e^{\sigma(i)})$. It follows
that $I_{f\circ \sigma}=\sigma^{-1}(I_f)$ and similarly
$O_{f\circ\sigma}=\sigma^{-1}(O_f)$. Then 
\[
p_{I_{f\circ \sigma}}=p_{\sigma^{-1}(I_f)}=p_{I_f}\circ\sigma\le f\circ\sigma\le
p_{O_f}^*\circ \sigma=p_{\sigma^{-1}(O_f)}^*=p_{O_{f\circ\sigma}}^*.
\]
For (iii), we have using properties of $\otimes$ and $\parr$ in Appendix \ref{app:fn},
together with Lemma \ref{lemma:causalp} (iii),
\[
p_I=p_{I_f}\otimes p_{I_g}\le f\otimes g \le f\vtl g, g\vtl f\le f\parr g\le  p^*_{O_f}\parr
p^*_{O_g}=p_O^*.
\]
The statement (iv) is quite clear, since the set of all subtypes with input/output sets $I/O$
is the interval $[p_I,p_O^*]$ in the distributive lattice $\Fe_n$.

\end{proof}

\subsection{Regular subtypes}

Let  $\Te_{n,O}$ denote the set of all type functions with
inputs/outputs in $I/O$. As we have seen, $\Te_{n,O}$ is contained in the set of all
subtypes with the same inputs/outputs, which form the interval $[p_I,p_O^*]$. The
(distributive) sublattice generated by $\Te_{n,O}$ is contained in the interval as well,
and by Lemma
\ref{lemma:Tf}, the set $T_1(f)$ for any function in the sublattice is obtained by taking
affine combinations of channels contained in intersections of higher order types. 
The aim of this paragraph is to characterize functions in this sublattice.

\begin{defi}\label{defi:regular_subtype} Let $O\subseteq [n]$. 
The sublattice generated by $\Te_{n,O}$  will be denoted by
$\Reg_{n,O}$ and  its elements  will be called \emph{regular subtypes} (with
outputs in $O$). The set of all regular subtypes in $\Fe_n$ will be denoted by $\mathcal
R_n$, that is, $\Reg_n=\bigcup_{O\subseteq [n]} \Reg_{n,O}$.
\end{defi}

By the above definition, we have 
$\Te_{n,O}\subseteq \Reg_{n,O}\subseteq [p_I,p_O^*]$. It is easily observed that in the extreme cases
$O=[n]$ or $O=\emptyset$, the above interval is a singleton, so both inclusions are
trivially  equalities.
As shown  in \cite{jencova2024onthestructure}, for $n=2$ both inclusions are equalities for any $O\subseteq
[2]$, but already for
$n=3$ both inclusions are strict in the nontrivial cases.

 Given a subset  $O\subseteq [n]$, we  introduce a partial order $\le_O$ in
 $\{0,1\}^n$ as $s\le_O s'$ if
$s_i\le s_i'$ for $i\in O$ and $s_i\ge_O s_i'$ for $i\notin O$. A subset of strings $D$ is
a downset (with respect to $\le_O$) if $s\in D$ and $s'\le_O s$ implies $s'\in D$, an upset is defined accordingly.
Note that for the complement $O^C=[n]\setminus O$ we get the opposite ordering:
$\le_{O^C}=\le_O^{op}$. A
principal  downset is a downset generated by a single element, we put $D^O_s:=\{s'\colon
s'\le_O s\}$. A principal upset $U^O_s$ is defined accordingly. 

We  say that a function $f\in \Fe_n$ is
 \emph{monotone} if it is monotone with respect to $\le_{O_f}$, that is, $s\le_{O_f} s'$ implies $f(s)\le f(s')$  (with the
usual  ordering $\le$ in $\{0,1\}$, i.e. $\le=\le_{[1]}$ in our notation above). 

\begin{lemma}\label{lemma:regular_alternative}
Let $f\in \Fe_n$ and put $O=O_f$.  The following are equivalent.
\begin{enumerate}
\item[(i)] $f$ is a monotone subtype;
\item[(ii)]  $f$ and $f^*$ are both monotone;

\item[(iii)] $\supp_f$ is an upset with respect to $\le_O$, such that $\supp_f\cap D^O_\theta=\{\theta\}$.

\end{enumerate}

\end{lemma}

\begin{proof}  Assume (i),  we will show that $f^*$
is monotone as well. Since $\le_{O_{f^*}}=\le_O^{op}$, we need to prove that $f^*$ is antitone
with respect to $\le_O$. So let $s'\le_O s$. If $s'\ne \theta\ne s$, then since $f$ is
monotone, we have $f^*(s')=1-f(s')\ge 1-f(s)=f^*(s)$, and clearly also $f^*(\theta)=1\ge
f^*(s)$, for any $s$. If $s=\theta$, then $s'\le \theta$ implies that $s'_i=0$ for all $i\in O$, so that $p_O(s')=1$.
Since $f$ is a subtype, $f^*$ is a subtype as well and we have $O=I_{f^*}$, so that
$f^*(s')\ge p_O(s')=1=f^*(s)$. This proves (ii). 

Assume (ii), then $\supp_f$ is clearly an upset  and
$\supp_{f^*}=\supp_f^C\cup\{\theta\}$ is a downset. This implies that
$D^O_\theta\setminus\{\theta\}\subseteq \supp_f^C$. Since $\theta\in \supp_f$, (iii) follows.

Finally, assume (iii), then $f$ is clearly monotone. To see that it is a subtype, we need
to show that $p_I\le f\le p_O^*$, which amounts to showing that for any string $s$,
$p_I(s)=1$ implies $f(s)=1$ and $p_O^*(s)=0$ implies $f(s)=0$. Now note
that $p_I(s)=1$ if and only if $\theta\le_O s$, that is, $s\in U^O_\theta$.  Since
$\supp_f$ is an upset containing $\theta$, this implies that  $U^O_\theta \subseteq
\supp_f$, hence  $p_I\le f$.
Further, $p_O^*(s)=0$ iff $s\ne \theta$ and $p_O(s)=1$, that is, $s\in
D^O_\theta\setminus\{\theta\}$. By the assumptions, this set is contained in $\supp_f^C$, hence $f(s)=0$. This shows that $f\le
p_O^*$.

\end{proof}

Let us list some further results on monotone subtypes.

\begin{lemma}\label{lemma:regular}
\begin{enumerate}
\item[(a)] If $f$ and $g$ are monotone subtypes with $O_f=O_g$, then $f\vee g$, $f\wedge
g$ are monotone as well.

\item[(b)] Let $f\in \Fe_n$ and  $g\in \Fe_m$ be  monotone subtypes and let $\sigma\in
\permut_n$ be any permutation.
Then $f\circ\sigma$, $f^*$, $f\vtl g, f\otimes g, f\parr g$
are monotone subtypes as well.

\item[(c)] Any type function is monotone.

\end{enumerate}

\end{lemma}

\begin{proof}  
The statement (a) is easily seen, since both the interval $[p_I,p_O^*]$ and the
monotonicity condition are closed under the lattice operations in $\Fe_n$. 
It is also quite clear that the set of monotone subtypes is closed under permutations. The
statement for complements follows by Lemma \ref{lemma:regular_alternative} (ii). 

Put $h:=f\vtl g$, then $h$ is a subtype by Proposition \ref{prop:subtype_dual_tensor}. Let  $t\le_{O_h}s$ and let 
$s=s^1s^2$ and $t=t^1t^2$ be the corresponding decompositions, then we have $t^1\le_{O_f}
s^1$, $t^2\le_{O_g}
s^2$. If $t^1=s^1$, then 
\[
h(t)=\begin{dcases} f(t^1)=f(s^1)=h(s) & \text{if } s^1\ne \theta\\
g(t^2)\le  g(s^2)=h(s) & \text{if } s^1=\theta.
\end{dcases}
\]
Assume that $t^1\lneq_{O_f} s^1$. Then $s^1=\theta$ implies $\theta\ne t^1\in D^{O_f}_\theta$,
in which case  $h(t)=f(t^1)=0\le h(s)$.   On
the other hand,  $t^1=\theta$ implies  $s^1\ne \theta$ and hence $h(s)=f(s^1)\ge f(\theta)= 1\ge h(t)$. Finally, if
$t_1\ne \theta\ne s^1$, then  $h(t)=f(t^1)\le f(s^1)=h(s)$ so that $f\vtl g$ is  monotone. 
The statements for $\otimes $ and $\parr$ follow by (a) and Lemma \ref{lemma:causalp}.
For (c), observe that quite trivially,  $1_1,p_1$ are monotone (sub)types. The statement now follows from
 (b) by induction.

\end{proof}

\begin{remark} One can check  that the monotonicity condition is equivalent to the
following pair of conditions:
\begin{enumerate}
\item[(a)] if $i\in O$ and $s\in \{0,1\}^{i-1}$, $s'\in \{0,1\}^{n-i}$ are such that
$f(s0s')=1$, then also $f(s1s')=1$,
\item[(b)]if $i\in I$ and $s\in \{0,1\}^{i-1}$, $s'\in \{0,1\}^{n-i}$ are such that
$f(s1s')=1$, then also $f(s0s')=1$.
\end{enumerate}
It was shown in \cite[Lemma 16]{apadula2024nosignalling} that (in our language) any type
function satisfies these conditions.

\end{remark}

Our goal in this paragraph is to prove the following statement.

\begin{theorem}\label{thm:regular} A subtype $f$ is regular if and only if it is monotone.

\end{theorem}

For the proof, we will need to study the structure of the lattice of monotone subtypes
with the same output space $O$, which for now we will denote by $\Me_{n,O}$. Clearly,
this is a distributive lattice, with smallest element $p_I$ and largest element
$p_O^*$, both of which are type functions. Notice that in the  cases $O=\emptyset$ or
$O=[n]$, the only subtype is $p_n$ resp. $1_n$, both of which are type functions, so that we may 
ignore these cases.

As we have seen, for a monotone subtype $f$, $\supp_f$ is an upset with
respect to the ordering $\le=\le_{O_f}$ and it is easily seen that $f\le g$ iff
$\supp_f\subseteq \supp_g$, and that  $\supp_{f\vee g}= 
\supp_f\vee \supp_g$, $\supp_{f\wedge g}=\supp_f\wedge \supp_g$. 
Let now  $s^1,\dots, s^k$ be the minimal elements in the upset
$\supp_f$. It is easily checked that $\supp_f$ is the union of the corresponding principal
upsets, $\supp_f=\cup_{j=1}^k U_{s^j}$. We therefore have
\[
U_{s^j}\cap D_\theta\subseteq \supp_f\cap D_\theta=\{\theta\}.
\]
It follows that $s^j\le \theta$ implies that $s^j=\theta$, for any $j$. Moreover, since
$\theta\in \supp_f$, there is some $j$ such that $\theta\in U_{s^j}$, in which case we
must have $s^j=\theta$. 

Note that for any $s\in \{0,1\}^n$, $s\nleq \theta$, we have  $U_{s}\cap
D_\theta=\emptyset$ 
and using Lemma \ref{lemma:regular_alternative}(iii), we see that 
$U_{s}\cup U_\theta$ is the support of some element of $\Me_{n,O}$, let us denote this
subtype by $f_s$. We obtain
\[
\supp_f=\cup_j U_{s^j}=\cup_{j, s^j\ne \theta} U_{s^j}\cup U_\theta=\cup_{j, s^j\ne \theta}\supp_{f_{s^j}},
\]
so that $f=\vee_{j, s^j\ne\theta} f_{s^j}$. It is therefore enough to consider only functions of the form
$f_s$ for  strings $s$ such that $s\nleq \theta$, that is, $s\in D_\theta^C$.  
Since $s\in U_\theta$ implies that $f_s=f_\theta=p_I$ is the least element, it is enough to consider
$s\in D_\theta^C\cap U_\theta^C=(D_\theta\cup U_\theta)^C$. In other words, these are
precisely the strings where there are some $i\in I$ and $j\in O$ such that $s_i=s_j=1$. 
Strings  $s\in (D_\theta\cup U_\theta)^C$ such that $s_j=1$ for exactly one $j\in 1$ and
$s_i=0$ for at most one $i\in I$ are called \emph{basic} (for the ordering $\le_O$).

\begin{lemma}\label{lemma:fsbasic} Let $s\in (D_\theta\cup U_\theta)^C$. Then there are
basic strings $t^k$ such that  $f_s=\wedge_k f_{t^k}$.
\end{lemma}

\begin{proof} It is quite obvious from the definition of the functions that for $s,t\in
(D_\theta\cup U_\theta)^C$, we have  $s\le t$ if and only if $f_t\le f_s$.
  Further,  for any  $f\in \Me_{n,O}$, we have $f(r)=1$ for $r\in U_\theta$ and
$f(r)=0$ for $r\in D_\theta\setminus \{\theta\}$, hence any monotone subtype is determined
by its values on strings   in $(D_\theta\cup U_\theta)^C$. For  $s,t,r \in
(D_\theta\cup U_\theta)^C$, we have 
\[
(f_s\wedge f_t)(r)=1\iff f_s(r)=f_t(r)=1\iff (t\le r) \& (s\le r) \iff t\vee s\le r \iff
f_{t\vee s}(r)=1,
\]
which implies that  $f_{s\vee t}=f_s\wedge f_t$. As can be easily seen from the definition
of $\le_O$, any $s\in (D_\theta\cup
U_\theta)^C$ can be written as $s=\vee_k t^k$, where $t^k$ are basic strings for $\le_O$.
This finishes the proof.

\end{proof}

Note that for $s\wedge t$ we only have the inequality $f_s\vee f_t\le f_{s\wedge t}$, the
opposite is not true apart from trivial cases. Indeed, note that $f_{s\wedge t}\le f_s\vee f_t$ implies that at
least one of 
$f_s$ or $f_t$ is equal to 1 at $s\wedge t$. But if, say, $f_s(s\wedge t)=1$, then either
$\theta\le s\wedge t$, so that $s,t\in U_\theta$, or $s\le s\wedge t$, which implies that
$s=s\wedge t$.

\begin{proof}[Proof of Theorem \ref{thm:regular}] The 'only if' part of the theorem follows easily from Lemma \ref{lemma:regular}. We now prove
the converse.  As noted above, we may assume that
$\emptyset \ne O\subsetneq [n]$. By Lemma \ref{lemma:fsbasic} and the paragraph above it,
the lattice $\Me_{n,O}$ is generated by functions of the form $f_s$ where $s$ is a basic
string for the ordering $\le_O$. We will prove that any such $f_s$ is a type function,
which will finish the proof. 

Assume that $i\in I$ is the unique  index  such that $s_i=0$. To ease the notations, we may  assume
that $i=1$, since both the type functions and monotone subtypes behave well under
permutations.   If  $r$ is  any string such that $r_1=1$, then $s\nleq r$ and
$\theta\nleq r$, so that we must have $f_s(r)=0$. It follows that
\[
f_s(t)=(1-t_1)f_s(0t_2\dots t_n)=p_1(t_1)g(t_2\dots t_n),\qquad \forall t\in \{0,1\}^n.
\]
It is straightforward that the function $g(r_1\dots r_{n-1}):=f(0r_1\dots r_{n-1})$ is in
$\Fe_{n-1}$ and  $j\in
O_g$ iff $j+1\in O_f$. Considering the corresponding ordering $\le_{O_g}$ in $\{0,1\}^{n-1}$
it  is easily checked that $g=f_{\tilde s}$ 
where $\tilde s\in \{0,1\}^{n-1}$ is the string $\tilde s=s_2\dots s_n$ (we slightly abuse
the notation here). Hence $f_s=p_1\otimes f_{\tilde s}$, so that $f_s$ is a type function
if $f_{\tilde s}$ is.  We may therefore restrict to strings with  $s_i=1$ for all $i\in I$ and $s_l=0$ for all  $l\in O$ 
except a single index $l=j$. We will
show that in that case 
\[
f_s=1-p_{\{j\}}+p_{I\cup\{j\}}
\]
which is a chain type for the chain $\emptyset\subsetneq \{j\}\subsetneq I\cup\{j\}$, so in
particular a type function. 

As previously observed, it is enough to compare the values of the functions for   $r\in (U_\theta\cup
D_\theta)^C$. In this case  $f_s(r)=1$ if and only if $s\le r$, which is equivalent to $r_j=1$. Since
these are precisely the strings for which the function on the right hand side is equal to 1, the statement
is proved.

\end{proof}

\begin{coro}\label{coro:normal} The lattice $\Reg_{n,O}$ is generated by chain types in
$\Te_{n,O}$, of the form 
\[
1-p_{\{j\}}+p_{I\cup\{j\}}, \quad p_{\{i\}}-p_{\{i,j\}}+p_{I\cup\{j\}},\qquad j\in O,\ i\in I.
\]
\end{coro}

\begin{proof}
 Straightforward from the proof of Theorem \ref{thm:regular}.
\end{proof}

We have also seen that the generating chain types have the form
\begin{center}
\vcbox{\includegraphics[scale=0.9]{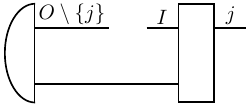}}
\vcbox{$\qquad \qquad$}
\vcbox{\includegraphics[scale=0.9]{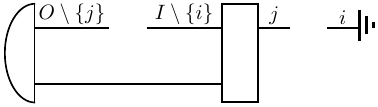}}
\end{center}
for some $i\in I$ and $j\in O$.

\begin{exm}\label{exm:reg_3} For $n=3$ and $O=\{1,3\}$, there are only 5
regular subtypes: beside the minimal element $p_{\{2\}}$ and the maximal element
$p_{\{1,3\}}^*$, there
are two chain types $\beta_1:= 1_3-p_{\{1\}}+p_{\{1,2\}}$ and
$\beta_2:=1_3-p_{\{3\}}+p_{\{2,3\}}$. The remaining element is $\beta_1\wedge \beta_2$. There are only two basic strings in this case: 110 and 011, and we have $f_{110}=\beta_1$,
$f_{011}=\beta_2$. 
\end{exm}

\begin{exm}\label{exm:reg_4}
For $n=4$, $O=\{1,3\}$, there are 50 regular subtypes and 14 chain types (including the minimal element
$p_{\{2,4\}}$ and the maximal element $p_{\{1,3\}}^*$), but only 6 basic strings: 1100,
1001, 0110, 0011, 1101, 0111. Any other chain type in $\Te_{4,\{1,3\}}$ is obtained as the join
of two of the  basic chain types. For example, the 2-comb
$\beta=1_4-p_{\{1\}}+p_{\{1,2\}}-p_{\{1,2,3\}}+p_4$
is obtained as $\beta=f_{1101}\vee f_{0011}$, where 
\[
f_{1101}=1_4-p_{\{1\}}+p_{\{1,2,4\}},\qquad f_{0011}=p_{\{2\}}-p_{\{2,3\}}+p_{\{2,3,4\}}.
\]

\end{exm}

\section{Signaling}\label{sec:signaling}

Let $A_1,\dots, A_n$ be elementary types and let $\mathcal H=\otimes_{i=1}^n\mathcal
H_{A_i}$, $\mathcal H_T=\otimes_{i\in T}\mathcal H_{A_i}$ for any $T\subseteq [n]$. 
Let $[n]=I\sqcup O$ and let $\Phi$ be a channel $B(\mathcal H_I)\to B(\mathcal H_O)$, with
Choi operator $C_\Phi\in B(\mathcal H)$.  We say that there is no
signalling from  $i\in I$ to $j\in O$ in $\Phi$, in notation $i\not \rightsquigarrow_\Phi j$, if 
\[
\Tr_{O\setminus \{j\}} [C_\Phi]=I_{\mathcal H_{A_i}}\otimes C_{\Phi_1},
\]
where $\Phi_1: B(\mathcal H_{I\setminus\{i\}})\to B(\mathcal H_{A_j})$ is some channel
(\cite{beckman2001causal,eggeling2002semicausal,apadula2024nosignalling}). 
In diagram
\begin{center}
\vcbox{\includegraphics[angle=90,origin=b]{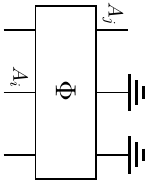}}
\vcbox{\hspace*{1em} $=$ \hspace*{1em}}
\vcbox{\includegraphics[angle=90,origin=b]{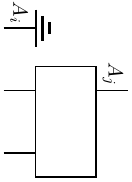}}
\end{center}
This is interpreted as the fact that the inputs at $i$ cannot influence the outcomes on
any measurements performed at the output $j$. Note also that the realization theorem for
quantum supermaps \cite{chiribella2009theoretical} implies that $\Phi$ is a quantum superchannel, more precisely, a
quantum higher order map of type $(A_j\to
A_i)\to (\otimes_{l\in I\setminus\{i\}}A_l\to \otimes_{j'\in O\setminus\{j\}}A_{j'})$.

We will consider the case when this property
holds for all  maps of a given regular subtype.

\begin{defi}
Let $[n]=I\sqcup O$ and let  $f\in \Reg_{n,O}$.  We say that there is no
signalling from $i\in I$ to $j\in O$ in $f$, in symbols: $i\not\rightsquigarrow_f j$, if 
$i\not \rightsquigarrow_\Phi j$ for any choice of
(nontrivial) elementary types $A_1,\dots,A_n$ and any channel $\Phi$ such that $C_\Phi\in
T_1(f)$. If this is not the case, we say that $i$ can signal to $j$ in
$f$, in symbols: $i\rightsquigarrow_f j$.

\end{defi}

We will now see that the relations $i\not\rightsquigarrow_f j$ in regular subtypes can be quite easily
characterized. For this, let us define the string $e^{i,j}\in \{0,1\}^n$ by
\begin{equation}\label{eq:eij}
e^{i,j}_l=1\ \iff \ l\in \{i,j\},\qquad l=1,\dots,n.
\end{equation}
We will sometimes use the notation $e^{i,j:n}$ to stress the value of $n$. 

\begin{theorem}\label{thm:signalling_f} Let $[n]=I\sqcup O$ and $f\in \Reg_{n,O}$. For
$i\in I$ and $j\in O$, we have $i\not\rightsquigarrow_f j$ if and only if $f(e^{i,j})=0$.

\end{theorem}

\begin{proof} As noted above, we have $i\not\rightsquigarrow_f j$ if and only if any
channel $\Phi$ with $C_\Phi\in T_1(f)$ is a superchannel with the  type structure $(j\to i)\to (
I\setminus \{i\}\to O\setminus\{j\})$. By Lemma \ref{lemma:Tf}, this is equivalent to $f\le \beta$, where 
\[
\beta= 1-p_{O\setminus \{j\}}+p_{(O\cup \{i\})\setminus\{j\}}-p_{O\cup \{i\}}+p_n
\]
is the  corresponding type function of the superchannels, see Example
\ref{exm:combs}. 

If $f\le \beta$, then we must
have $f(e^{i,j})\le \beta(e^{i,j})=0$, which is easily seen by plugging $e^{i,j}$ into the
expression for $\beta$. For the converse, assume that $f(e^{i,j})=0$. We have
to show that $f(s)=0$ whenever $\beta(s)=0$ for $s\in \{0,1\}^n$.    From the chain of
inequalities $p_n\le p_{O\cup\{i\}}\le p_{(O\cup\{i\})\setminus \{j\}}\le
p_{O\setminus\{j\}}$, we see that for
$\beta(s)=0$, we must have $p_{O\setminus\{j\}}(s)=1$ and either $p_{(O\cup\{i\})\setminus
\{j\}}(s)=0$,  or 
$p_{O\cup\{i\}}(s)= p_{(O\cup\{i\})\setminus\{j\}}(s)=1$ and $p_n=0$. 

In the latter case, we have 
$f(s)\le p_O^*(s)\le p_{O\cup\{i\}}^*(s)=0$, here the first inequality holds since $f$ is a subtype
with outputs in $O$. Let now $p_{O\setminus\{j\}}(s)=1$ and $p_{(O\cup\{i\})\setminus\{j\}}(s)=0$. This
implies that $s_l=0$ for all $l\in O$, $l\ne j$, and $s_i=1$. It follows that 
$s_l\le e^{i,j}_l$ for all $l\in O$, and $s_l\ge e^{i,j}_l$ for all $l\in I$, in other
words, $s\le_O e^{i,j}$. Theorem \ref{thm:regular} now implies that $f(s)\le
f(e^{i,j})=0$. This finishes the proof.

\end{proof}

The following behaviour of the signalling relations should be clear from the
interpretation of  the basic operations on subtypes, but is also easily proved from
Theorem \ref{thm:signalling_f}. 

\begin{coro}\label{coro:signalling_operations} Let $f\in \Reg_{n,O}$, $i\in I$, $j\in O$.
Then
\begin{enumerate}
\item[(i)]  $i\not\rightsquigarrow_f j$ $\iff$ $j\rightsquigarrow_{f^*}i$.
\item[(ii)] For any $\sigma\in \permut_n$, $i\not\rightsquigarrow_f j$ $\iff$ $\sigma^{-1}(i)\not\rightsquigarrow_{f\circ\sigma} \sigma^{-1}(j)$.
\item[(iii)] Assume the decomposition $[n]=[n_1]\oplus[n_2]$ and let $f=f_1\otimes f_2$,
with $f_1\in \Reg_{n_1,O_1}$, $f_2\in \Reg_{n_2,O_2}$. For  $i_k\in
I_k$, $j_k\in O_k$, $k=1,2$, we have 
\begin{itemize}
\item $i_1\not\rightsquigarrow_{f_1} j_1$ $\iff$  $i_1\not\rightsquigarrow_f j_1$,
\item $i_2\not\rightsquigarrow_{f_2} j_2$ $\iff$  $(i_2+n_1)\not\rightsquigarrow_f
(j_2+n_1)$ 
\item $i_1\not\rightsquigarrow_f (j_2+n_1)$,  $(i_2+n_1)\not\rightsquigarrow_f j_1$.
\end{itemize}

\item[(iv)] With the same notations  as in (iii), let  $f = f_1\vtl f_2$.  
Then 
\begin{itemize}
\item $i_1\not\rightsquigarrow_{f_1} j_1$ $\iff$ $i_1\not\rightsquigarrow_f j_1$,
\item $i_2\not\rightsquigarrow_{f_2} j_2$ $\iff$$(i_2+n_1)\not\rightsquigarrow_f
(j_2+n_1)$
\item $i_1\not\rightsquigarrow_f (j_2+n_1)$ but  
$(i_2+n_1)\rightsquigarrow_f j_1$.
\end{itemize}

\item[(v)] For any  $g\in \Reg_{n,O}$, 
\begin{itemize}
\item $i\not\rightsquigarrow_{f\wedge g} j$ $\iff$ $i\not\rightsquigarrow_f j$ or
$i\not\rightsquigarrow_g j$.
\item 
$i\not\rightsquigarrow_{f\vee g} j$ $\iff$  $i\not\rightsquigarrow_f j$ and 
$i\not\rightsquigarrow_g j$. 
\end{itemize}

\end{enumerate}

\end{coro}

\begin{proof} The statement (i) follows from $f(e^{i,j})=0$ if and only if
$f^*(e^{i,j})=1-f(e^{i,j})=1$. For (ii), note first that by Proposition
\ref{prop:subtype_dual_tensor} we have $I_{f\circ \sigma}=\sigma^{-1}(I)$,
$O_{f\circ\sigma}=\sigma^{-1}(O)$. The statement follows from 
$f(e^{i,j})=f\circ\sigma(e^{\sigma^{-1}(i),\sigma^{-1}(j)})$.

The statements (iii) and (iv) follow from expressing the corresponding strings as
concatenations:
\begin{align*}
e^{i_1,j_1:n}& =e^{i_1,j_1:n_1}\theta_{n_2},\quad 
e^{(i_2+n_1),(j_2+n_1):n}=\theta_{n_1}e^{i_2,j_2:n_2},\\
e^{i_1,(j_2+n_1):n}&=e^{i_1:n_1}e^{j_2:n_2},\quad  e^{(i_2+n_1),j_1:n}=e^{j_1:n_1}e^{i_2:n_2}.
\end{align*}
We then use the definitions of $\otimes$ and $\vtl$, together with Theorem
\ref{thm:signalling_f} and Proposition \ref{prop:inout}. The statement (v) is immediate
from Theorem \ref{thm:signalling_f}.

\end{proof}

\subsection{Signalling  in higher order types }

The aim of this paragraph is to show that signalling relations for a type can be inferred from the order
properties of its structure posets. 
The corresponding Hasse diagrams with labelled vertices 
then provide a  graphical representation of the signalling relations.

We have seen in Section \ref{sec:structure} that we
may use the rank function in $\Pe_f$ to define the rank $r_f(i)$ of an index $i\in [n]$,
and that this rank  determines the inputs/outputs  of $f$.  We now define the rank in $f$ for pairs of indices.

\begin{defi}\label{defi:rank} Let $f\in \Te_n$ and $i,j\in [n]$. We define the \emph{rank}
of the pair $i,j$ in $f$ as
\[
r_f(i,j):= \max\{\rho_f(S\wedge T)\colon S\in \TT_i^f,\ T\in \TT_j^f\}.
\]
If the set on the right hand side is empty, this value is set to -1, unless $i$ or $j$ is
a free output, in which case we put $r_f(i,j)=\min\{r_f(i),r_f(j)\}$. 

\end{defi}

It is easy to see by the properties of the rank function  that $r_f(i,j)\le \min\{r_f(i),r_f(j)\}$, 
with equality if and only if $i$ or $j$ is a free output or there exist 
some $S\in \TT_i^f$ and $T\in \TT_j^f$ that are comparable. Further, by Corollary
\ref{coro:meet}, if $S\wedge T$ exists, then it belongs to $\Pe_f^0$ and hence it is the
meet of $S$ and $T$ also in $\Pe_f^0$, moreover, by Corollary \ref{coro:lattice} such a
meet exists if and only if $S$ and $T$ have a common lower bound in $\Pe_f^0$.

\begin{theorem}\label{thm:signalling_pf}
Let $f\in \Te_n$ and $i\in I_f$, $j\in O_f$. Then $i\not\rightsquigarrow_f j$ if and only if
  $r_f(i,j)$ is even.

\end{theorem}

\begin{proof} 
Let us first note that the statement is true if $i\in I_f$  is any input and $j\in O_f$ is
a free output of $f$. Indeed, in this case, $j$ does not belong to any $S\in \Pe_f$, and
hence the value of $s_j$ does not influence the value of $f(s)$ for any string $s$. In
particular, $f(e^{i,j})=f(e^i)=0$, so $i\not\rightsquigarrow_f j$. On the other hand, we
have $r_f(i,j)=\min\{r_f(i),r_f(j)\}=r_f(i)$ is even, by definition of the ranks. Hence we
may assume below that $j$ is not a free output.

We will proceed by induction. Assume first that $f$ is a chain type, with 
$S_0\subsetneq \dots \subsetneq S_K\subseteq [n]$ the corresponding chain of subsets. 
It is easily seen that in this case $k:=r_f(i,j)=\min\{r_f(i),r_f(j)\}$, and $S_k$ is the
smallest element in the chain that contains $i$ or $j$. Consequently,
\[
f(e^{i,j})=\sum_{l=0}^{k-1} (-1)^l,
\]
which is zero if and only if $k$ is even. 

Assume next that the statement holds for $f\in
\Te_n$ and let $\sigma\in \permut_n$. By \cite[Coro.~4.7]{jencova2024onthestructure},  
$\Pe_{f\circ \sigma}$ as a poset is the same
 as $\Pe_f$, only with label sets exchanged as $L_S\mapsto \sigma^{-1}(L_S)$. This shows that
 $r_{f\circ\sigma}(\sigma^{-1}(i),\sigma^{-1}(j))=r_f(i,j)$. Using Corollary
 \ref{coro:signalling_operations} (ii), we see that the statement holds for
 $f\circ\sigma$ as well.

Let us now look at $f^*$. Assume  $i\in I_{f^*}=O_f$ and $j\in O_{f^*}=I_f$. By Corollary
\ref{coro:signalling_operations} (i) we  have  $i\not\rightsquigarrow_{f^*} j$ if and only if $j\rightsquigarrow_f i$. 
By  the proof of \cite[Prop.~4.8]{jencova2024onthestructure}, we have $\TT^f_i=\TT^{f^*}_i$
and $\TT^f_j=\TT^{f^*}_j$, unless $j$ is a free output of $f^*$, which case was treated
in the first paragraph of this proof, or $i$ is a free output of $f$, which is not possible since $j\rightsquigarrow_f
i$. It follows that all elements labelled by $i$ and $j$ are the same in both $\Pe_f^0$
and $\Pe_{f^*}^0$, together with all the elements below them except $\emptyset$. 

By the assumption on $f$, $j\rightsquigarrow_f i$ is equivalent to the fact that $r_f(i,j)$
is  odd. If $r_f(i,j)=-1$, then  $S\wedge T$ does not exist 
for any $S\in \TT^f_i$ and $T\in \TT^f_j$. By Corollary \ref{coro:lattice}, this implies
that no such pair of $S$ and $T$ can have a lower bound in $\Pe_f^0$, in particular,
$\emptyset \notin \Pe_f^0$. But then $\emptyset \in \Pe_{f^*}^0$ becomes the largest lower
bound for any such pair in $\Pe_{f^*}^0$, and hence $r_{f^*}(i,j)=\rho_{f^*}(\emptyset)=0$. 
In the case that $r_f(i,j)$ is a positive odd integer, let $S\in \TT^f_i$ and $T\in \TT^f_j$ be such
that $r_f(i,j)=\rho_f(S\wedge T)$, then it follows by the above remarks that
$r_{f^*}(i,j)=\rho_{f^*}(S\wedge T)=\rho_f(S\wedge T)\pm 1$. In any case, we see that
$r_{f^*}(i,j)$ is even iff $r_f(i,j)$
is odd.

Finally, assume that the statement is true for $f_1\in \Te_m$ and $f_2\in \Te_n$ and let
$f=f_1\otimes f_2$. Let $i\in I_f$ and $j\in O_f$. Assume that $i$ is an index related to
$f_1$ and $j$ related to $f_2$, that is, $i\in I_{f_1}$ and $j-m\in O_{f_2}$. By Corollary
\ref{coro:signalling_operations} (iii), we have $i \not\rightsquigarrow_f j$.
As before, we may assume that $j$
is not a free output of $f$, so that $\TT^j_f\ne\emptyset$. Elements  $T\in \TT^f_j$ are
precisely those of the form  $T=(U_1,T_2)$, where $U_1\in \Min(\Pe_{f_1})$ and $T_2\in \TT^{f_2}_{j-m}$. 
Similarly, $S\in \TT^f_i$ if and only if $S=(S_1,V_2)$, where $S_1\in \TT^{f_1}_i$ and
$V_2\in \Min(\Pe_{f_2})$. Two such elements have a common lower bound if and only if
$U_1\le S_1$ and $V_2\le T_2$, in which case $S\wedge T=(U_1,V_2)\in \Min(\Pe_f)$. It
follows that $r_f(i,j)=0$. The case when $i$ is related to $f_2$ and $j$ to $f_1$ is
similar.

Assume that both $i$ and $j$ are related to the same function, say $f_1$, so that $i,j\in
[m]$. Then $i\not\rightsquigarrow_f j$ if and only if $i\not\rightsquigarrow_{f_1} j$. 
Furthermore, $S\in \TT^f_i$ and $T\in \TT^f_j$ have the form $S=(S_1,U_2)$, $T=(T_1,V_2)$, with $S_1\in
\TT^{f_1}_i$, $T_1\in \TT^{f_1}_j$ and $U_2,V_2\in \Min(\Pe_{f_2})$. Such elements can have a
common lower bound only if $U_2=V_2$, in which case $S\wedge T=(S_1\wedge T_1,U_2)$ and
$\rho_f(S\wedge T)=\rho_{f_1}(S_1\wedge T_1)$. It follows that if the statement holds for
$f_1$, it also holds for $f$.  The case that
both $i$ and $j$ are related to $f_2$ is treated in an analogous way.

\end{proof}

Although  the full structure poset can be
inferred from the reduced one, it might be not
obvious how to obtain the  rank of a labelled  element just from $\Pe_f^0$.
The next result shows an easy  way to infer signalling from
the reduced poset.

\begin{coro}\label{coro:signalling} Let $f\in \Te_n$, $i\in I_f$, $j\in O_f$. Then
\begin{enumerate}
\item[(i)] If $j$ is a free output, then $i\not\rightsquigarrow_f j$.
\item[(ii)] If there are some $S\in \TT_i^f$ and $T\in \TT_j^f$ that are comparable, then
$i\not\rightsquigarrow_f j$ if and only if $S\le T$.
\item[(iii)] If none of the above holds, then $i\not\rightsquigarrow_f j$ if and only if 
the longest chain below a pair $(S,T)$ with  $S\in \TT^f_i$ and $T\in \TT^f_j$ has even length (where we
set the length of an empty chain to -1).

\end{enumerate}

\end{coro}

\begin{proof} The statement (i) was shown in the proof of Theorem \ref{thm:signalling_pf}. 
If there are some comparable $S\le \TT_i^f$ and $T\le \TT_j^f$, then
$r_f(i,j)=\min\{r_f(i),r_f(j)\}=\min\{\rho_f(S),\rho_f(T)\}$. Since $i$ is an input and
$j$ an output, the last value is even if and only if
it is equal to $\rho_f(S)$, that is, $S\le T$. 
Under the assumptions in (iii), both $\TT^f_i$ and $\TT^f_j$ are nonempty and any pair 
$(S,T)$ with 
$S\in \TT^f_i$ and $T\in \TT^f_j$ is incomparable. By combination of  Lemma \ref{lemma:ypsilon} and Corollary 
\ref{coro:xfree}, we see
that if $S\wedge T$ exists, then the downset
$(S\wedge T)^{\downarrow_f}=S^{\downarrow_f}\cap T^{\downarrow_f}$ is a chain and $\rho_f(S\wedge T)$ is its length.
 
\end{proof}

\begin{exm}\label{exm:signalling_nspm}  We first check the signalling conditions in the
diagrams in Examples \ref{exm:combs} and \ref{exm:nspm}. For the chain types this is easy:
signalling is only allowed down the chain (Corollary \ref{coro:signalling}(ii)).
Similarly, for the nonsignalling channels type
$f_{ns}$ in
Example \ref{exm:nspm}, we have signalling down the chains $2\rightsquigarrow_{f_{ns}} 1$
and $4\rightsquigarrow_{f_{ns}} 3$, but
$2\not\rightsquigarrow_{f_{ns}} 3$ and $4\not \rightsquigarrow_{f_{ns}} 1$, since the only common lower
bound for 2,3 resp 1,4 is the least element $\emptyset$, so that
$r_{f_{ns}}(2,3)=r_{f_{ns}}(1,4)=0$. 

In the case of the process matrices, we see that $2\not\rightsquigarrow_{f_{pm}} 1$ and $4\not
\rightsquigarrow_{f_{pm}} 3$ since signalling is not allowed up the chains, but since there is no
common lower bound for $2,3$, resp. $4,1$, we have $r_{f_{pm}}(2,3)=r_{f_{pm}}(4,1)=-1$, so that
signalling is possible. In fact, note that the process matrix type has the same signalling
relations as the permuted nonsignalling channel type $\tilde f_{ns}$ depicted below. 
\begin{center}
\begin{minipage}{0.3\textwidth}
\centering
\includegraphics[scale=0.9]{pm_hasse.pdf}

$f_{pm}$
\end{minipage}%
\begin{minipage}{0.3\textwidth}
\centering
\includegraphics[scale=0.9]{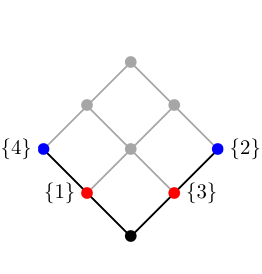}

$\tilde f_{ns}$
\end{minipage}%
\begin{minipage}{0.3\textwidth}
\centering
\includegraphics[scale=0.9]{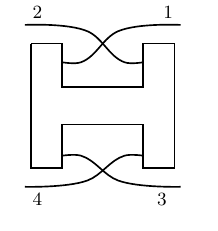}

process matrix as a ns channel
\end{minipage}%
\end{center}
One can check that $f_{pm}\le \tilde f_{ns}$ and the values of the functions differ only at
a single string $s=1111$, this also confirms that the signalling relations must be the same. 
This inclusion is obtained by presenting the process matrix as a
nonsignalling channel by 'extending the input/output wires'.
However, this changes the ordering 'down the chain' of the indices  1,2, resp. 3,4
in the Hasse diagram of $\Pe^0_{f_{pm}}$.

\end{exm}

\begin{exm}[Adapters]\label{exm:adapters}  Adapters transform process matrices to process matrices. If
$g=f_{pm}$ is the type function for process matrices as in Example \ref{exm:nspm}, then the corresponding
adapter has type function $a_1:=(g\otimes g^*)^*\in \Te_8$. If $k$ is the type function for
process matrices  with a global past and future as in Example \ref{exm:pm_pf}, the type function
for the adapter is $a_2=(k\otimes
k^*)^*\in \Te_{12}$. The structure posets $\Pe_{a_1}$ and $\Pe_{a_2}$ are too large, we show only
the reduced posets, which nevertheless contain all the information about the types. 
The first diagram shows the adapter $a_1$, with the input process matrix related to
indices $1,\dots,4$  and the output process matrix with indices
$5,\dots,8$. The second diagram depicts the adapter $a_2$, with indices $1,\dots,6$
corresponding to  the input process matrix  and indices $7,\dots,12$ relating to the
output process matrix. 
Note
also that while the two adapters have a similar definition, the corresponding posets have
quite  different structure:

\begin{center}
\vcbox{\includegraphics[scale=0.9]{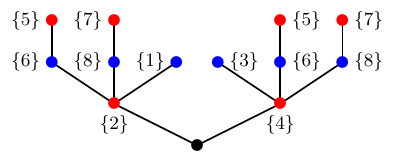}}
\vcbox{$\qquad$ $\qquad$}
\vcbox{\includegraphics[scale=0.9]{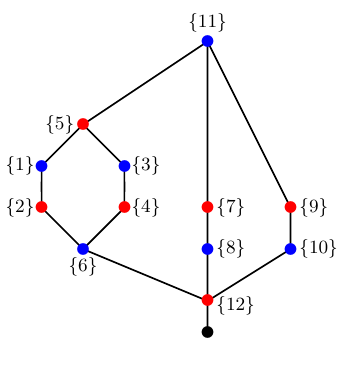}}

\end{center}
\medskip
Let us check the signalling relations in $a_1$ from  the first  diagram. Here the input
indices are 1, 3, 6 and 8, labelling blue vertices. Notice that the (output) index 2 labels the smallest element in the left
branch of the diagram that has a nonempty label, so all input indices in this branch (1,
6, 8) can signal to 2. This
element also provides the largest lower bound for the elements labelled by 1 and 5, resp.
1 and 7, so that 1 can signal to both 5 and 7.
Similarly, the input indices in the right branch (3, 6, 8) can signal to 4 and we have
$3\rightsquigarrow_{a_1} 5$ and $3\rightsquigarrow_{a_1} 7$. But we have
$1\not\rightsquigarrow_{a_1} 4$ and $3\not\rightsquigarrow_{a_1} 2$, because in both these cases
the only lower bound is the least element, so that $r_{a_1}(1,4)=r_{a_1}(3,2)=0$. 
 We also infer that $6\rightsquigarrow_{a_1} 7$ but $6\not\rightsquigarrow_{a_1} 5$, and
 similarly   $8\rightsquigarrow_{a_1} 5$ but $8\not\rightsquigarrow_{a_1} 7$.  

In the second diagram, corresponding to $\Pe_{a_2}^0$, the input indices are 1, 3, 6, 8,
10 and 11. Note that 11 is the label of a  largest element and can signal to any output. 
On the other hand, 12 labels the smallest element $S$ with a nonempty label set, so any
input can signal to 12.
In fact, for the three main branches in the diagram that depart at $S$, all inputs on a
branch can signal to any output on the other branches, but cannot signal to any of
the outputs above it.  On the other hand, the input 1 can signal to 2, since
$r_{a_2}(1,2)=r_{a_2}(2)=3$, but cannot signal to 4, since
$r_{a_2}(1,4)=r_{a_2}(6)=2$, similarly, $3\rightsquigarrow_{a_2} 4$ but  $3\not\rightsquigarrow_{a_2} 2$.  

\end{exm}

\subsection{The normal form of higher order types}

It was proved in \cite[Thm.~2]{hoffreumon2026projective} that any projection superoperator
generated from basic superoperator projections  using tensor products and duals 
has  a \emph{normal form}, which means that it can be expressed using only the one-way
signalling
compositions and lattice operations. In the special case of  higher order types, this
means that the corresponding projection for any type  can be expressed in terms of joins
and meets of projections related to causally
ordered types.  

In the language of type functions, this can be formulated as the fact that  $f\in \Te_{n,O}$
can be obtained from chain types in $\Te_{n,O}$ and lattice operations. It might be
tempting to interpret such a decomposition as a means to see whether the type is causally
ordered, or to detect possibility of indefinite causal order.
In general, however, such an expression is highly nonunique and it may be difficult to see
which expressions actually give the same type function. We have seen a decomposition of a type
function into joins and meets of chain types of a special form in
Corollary \ref{coro:normal}.  Note that  Example \ref{exm:reg_3} shows a
chain type that may be itself decomposed as the join of two basic chain types, so that a
normal form with a join does not exclude causally ordered types.

A version of the normal form  was also proved in the structural theorem \cite[Thm.~4.16]{jencova2024onthestructure}. It was argued there that
the basic orderings in  a given type can be obtained from the subchains of its  reduced structure
poset, from which a normal form can be constructed,  but the procedure is not straightforward. In this paragraph, we show that the
number of chain types in a normal form may be restricted by the reduced structure poset.

\begin{prop}\label{prop:normalform} Let $f$ be a type function and let $\Pe_f^0$ have $M$
maximal chains. Then there is a normal form of $f$ such that the number of different involved chain types 
is less or equal to $M$. 

\end{prop}

\begin{proof} We will apply the usual induction argument. The statement is trivial for
chains, and the property is obviously preserved by any permutations. It can be also seen
that $\Pe_{f^*}^0$ has $M$ maximal chains as well and since the complement of a chain type is
again a chain type, we conclude that if the statement is true for $f$, then it holds also
for $f^*$. It is therefore enough to look at the case when $f=f_1\otimes f_2$ for some
type functions $f_1\in \Te_{n_1,O_1}$ and $f_2\in \Te_{n_2,O_2}$, $O=O_1\oplus O_2$,  which satisfy the property. 
More precisely, we assume that for  $k=1,2$,
$\Pe_{f_k}^0$ has $M_k$ maximal chains and we may write 
\[
f_k=\bigvee_{a} \bigwedge_{b} f^k_{a,b}
\]
where $f^k_{a,b}\in \{\beta^k_1,\dots,\beta^k_{n_k}\}$,  $\beta^k_l\in \Te_{n_k,O_k}$  are chain types and $n_k\le M_k$. Let $\alpha_k\in
\Te_{n_k,O_k}$ be any chain type such that $f_k\le \alpha_k$ (such a chain type always
exists, we may for example take $\alpha_k=p_{O_k}^*$). Using Lemmas \ref{lemma:causalp}
and 
\ref{lemma:causalp_ext}, we see that
\[
f=f_1\otimes f_2= (f_1\vtl \alpha_2)\wedge
(f_2\vtl\alpha_1)=(\bigvee_{a}\bigwedge_{b}f^1_{a,b}\vtl \alpha_2)\wedge (\bigvee_{a'}\bigwedge_{b'}f^2_{a',b'}\vtl \alpha_1)
\]
Using distributivity of $\Fe_n$, we obtain an expression of the form
\[
f=\bigvee_{c}\bigwedge_{d} g_{c,d},
\]
where  $g_{c,d}\in \{\beta^1_1\vtl \alpha_2,\dots,\beta^1_{n_1}\vtl\alpha_2,
\beta^2_1\vtl\alpha_1,\dots,\beta^2_{n_2}\vtl\alpha_1\}$. Since $\beta^1_l\vtl \alpha_2$
and $\beta_l^2\vtl\alpha_1$ are chain types in $\Te_{n,O}$, we obtained an expression of $f$ in a normal form using 
$n_1+n_2$ chain types. Now note that any maximal chain in $\Pe_f^0$ is of the
form $\mathcal C =(\mathcal C_1,U_2)$ with $\mathcal C_1$ a maximal chain in $\Pe_{f_1}^0$
and $U_2\in \Min(\Pe_{f_2}^0)$, or $\Ce=(U_1,\Ce_2)$ for a maximal chain $\Ce_2$ in
$\Pe_{f_2}^0$ and $U_1\in \Min(\Pe_{f_1}^0)$. This shows that $n_1+n_2\le M_1+M_2\le M$.

\end{proof}

The normal form obtained in \cite{jencova2024onthestructure} has a specific 'minimax'
property: namely there are chain types  $f_{a,b}\in \Te_{n,O}$ such that 
\begin{equation}\label{eq:minimax}
f=\bigvee_{a\in A}\bigwedge_{b\in B} f_{a,b}=\bigwedge_{b\in B}\bigvee_{a\in A}f_{a,b}
\end{equation}
for some finite index sets  $A$, $B$. Such an
expression can be inferred from the construction of the type functions using the tensor
product and complement, together with the properties of the causal product in Section
\ref{sec:causal}. We will show in Appendix \ref{app:normal} that for type functions in Examples \ref{exm:nspm},
\ref{exm:pm_pf} and \ref{exm:adapters}, a normal form with this property can be obtained
by considering the maximal chains in the reduced structure posets,  and by  signalling
relations between elements of the chains, easily determined by the rank of their branching points.

\section{Conclusions}

We have investigated the structure of higher order types from an algebraic and  order-theoretic
perspective, using the combinatorial characterization in terms of boolean functions. 
We studied the lattice of regular subtypes, generated by  higher order types with fixed
sets of 
input and output indices,  and found their characterization by a monotonicity condition. 
We also studied signalling conditions for the corresponding maps. For higher order types,
we have shown that the order in the structure posets reveals the signalling relations and
we demonstrated on examples that the posets may be used to construct normal forms.  
A more detailed investigation of the relation of normal forms and structure posets is left
for future work. 

There are other possible further research directions. As it was proved in
\cite{apadula2024nosignalling}, the no-signalling relations constraints the possibility of
composition of higher order maps. Such a composition may not belong to a higher order
type, but it should be of a regular subtype. An interesting problem is the
connection to some constructions on the structure posets. It is a natural question whether there are similar structural posets for any regular
subtype, constructed from the M\"obius transform. Indeed, the transform exists for any
function, but the fact that the posets uniquely represent the type functions is based on
the fact that the M\"obius transform only attains values in the set $\{-1,0,1\}$. While
this remains true for all regular subtypes with $n=2,3$, there are examples with $n=4$
attaining values $\pm2$ or 3.  Nevertheless, it may still hold for regular subtypes
obtained by type compositions. A description  of composition of types in terms of
the posets and their labeled diagrams would be useful for understanding the structure of
higher order maps. However, such a description in itself would not be enough to capture
 properties such as dynamical causal order \cite{wechs2021quantum}.

As pointed out already in \cite{bisio2019theoretical}, the theory of higher order types is
in principle applicable beyond the quantum theory. In \cite{jencova2024onthestructure},
it was shown that the type functions can be applied to construct sets of higher order maps
over vector spaces with fixed affine subspaces. Under suitable admissibility conditions,
and proper choices of positive cones, such a construction may be used to build a theory of
higher order maps in a general probabilistic theory \cite{plavala2023general}. See also
\cite{bavaresco2024indefinite,sengupta2024achieving} for some constructions.

\subsection*{Acknowledgements}

I am indebted to Gejza Jen\v{c}a and Timoth\'ee Hoffreumon for discussions and useful
comments, and to Adam Jen\v{c}a for technical support. The work was supported by 
the grants VEGA 2/0128/24 and by the Science and Technology Assistance Agency under the contract
No. APVV-22-0570.

\subsection*{Declaration of generative AI use}

During the preparation of this work the author used Claude Sonet 4.6 for finding and
checking examples and counterexamples of boolean functions (Example \ref{exm:reg_4},
Appendix \ref{app:normal}). The author  takes full responsibility for the
content of the published article.

\appendix
\numberwithin{equation}{section}

\section{Boolean functions}\label{app:fn}
We collect some basic notations, definitions and properties related to binary strings and
boolean functions. See \cite[Appendix A]{jencova2024onthestructure} for more
information. 

\subsection{Basic definitions}
For $m\le n\in \mathbb N$, we will denote the corresponding interval $\{m,m+1,\dots,n\}$ by
$[m,n]$. For $m=1$, we will simplify to  $[n]:=[1,n]$. 
For $n_1,n_2\in \mathbb N$, $n_1+n_2=n$,
we will denote by $[n]=[n_1]\oplus [n_2]$ the decomposition of $[n]$ as a concatenation of two
intervals $[n]=[n_1][n_1+1,n_1+n_2]$.
Similarly, for $n=\sum_{j=1}^kn_j$, we have the decomposition
$[n]=\oplus_{j=1}^k[n_j]=[n_1][m_1+1,m_1+n_2]\dots[m_{k-1}+1,m_{k-1}+n_k]$,
where $m_j:=\sum_{l=1}^{j} n_j$. Note that the order of $n_1,\dots, n_k$ in
this decomposition is fixed.

The set of all subsets of $[n]$ will be denoted by $2^n$. 
With the inclusion ordering and complementation $S^C:=[n]\setminus S$,
$2^n$ is a boolean algebra, with smallest element $\emptyset$ and largest element
$[n]$.

A binary string of length $n$ is a sequence  $s=s_1\dots s_n$, where $s_i\in
\{0,1\}$. Such a string can be interpreted as an element in $\{0,1\}^n$, but also as a 
map $[n]\to \{0,1\}$, $i\mapsto s_i$. The string with all $s_i=0$ will be denoted by
$\theta_n$, the index $n$ will be dropped if it is clear.

The permutation group $\permut_n$ has an obvious action on $\{0,1\}^n$ defined  by
precomposition: for $\sigma\in \permut_n$ we define 
\[
\sigma(s):=s\circ\sigma^{-1}=s_{\sigma^{-1}(1)}\dots s_{\sigma^{-1}(n)},\qquad s\in
\{0,1\}^n.
\]
 For a decomposition $[n]=\oplus_{j=1}^k[n_j]$, we have a corresponding decomposition of
any string $s\in \{0,1\}^n$ as a concatenation of strings
\[
s=s^1\dots s^k,\qquad s^j\in \{0,1\}^{n_j},\qquad s^l_j:=s_{m_{l-1}+j},\ j\in[n_l],\ 
m_l=\sum_{i=1}^{l}n_i.
\]

\subsection{The boolean  algebra $\Fe_n$}
 We will deal with the subset of boolean functions defined as 
\[
\Fe_n:=\{f:\{0,1\}^n\to \{0,1\}\colon f(\theta)=1\}.
\]
We can see that  $\Fe_n$ is a boolean algebra, whose operations are the pointwise minimum $f\wedge g$ and maximum
$f\vee g$. The top element in $\Fe_n$ is the constant unit function $1_n(s)\equiv 1$,
 the  bottom element is 
\[
p_n(s):=\begin{dcases} 1 & \text{ if }s=\theta\\
0 &\text{ otherwise.}
\end{dcases}
\]
The complement operation is defined as $f^*=1-f+p_n$. In particular,  $\Fe_n$ is a finite distributive
lattice,  and the complementation satisfies $f^{**}=f$ and De Morgan laws: $(f\vee
g)^*=f^*\wedge g^*$, $(f\wedge g)^*=f^*\vee g^*$.

Given $m,n\in \mathbb N$ and a fixed decomposition $[m+n]=[m]\oplus [n]$, we will define
two product 
operations $\Fe_m\times \Fe_n \to \Fe_{m+n}$. For any string $s\in \{0,1\}^{m+n}$, let $s=s^1s^2$
be the decomposition as concatenation of strings $s^1\in \{0,1\}^m$, $s^2\in \{0,1\}^n$.
For $f\in \Fe_m$, $g\in \Fe_n$, we put:
\begin{align*}
(f\otimes g)(s)&=f(s^1)g(s^2),\quad s\in \{0,1\}^{m+n}\\
(f\parr g)&= (f^*\otimes g^*)^*.
\end{align*}

\begin{lemma}\label{lemma:products} Let $f\in \Fe_m$, $g\in \Fe_n$. Fix
the decomposition $[m+n]=[m]\oplus [n]$. Then 
\begin{enumerate}
\item[(i)] $f\otimes g\le f\parr g$.
\item[(ii)] With $\tilde f\in \Fe_m$, $f\le \tilde f$ and $\tilde g\in \Fe_n$, $g\le
\tilde g$, we have
\[
f\otimes g\le \tilde f\otimes \tilde g,\qquad f\parr g\le \tilde f\parr \tilde g.
\]

\end{enumerate}

\end{lemma}

\begin{proof} The proof of  (i) can be found in \cite[Prop.~4.15]{jencova2024onthestructure}. The
first inequality in (ii) is immediate from the definition of $f\otimes g$, the second
inequality follows from the first by duality.

\end{proof}

\section{The structure  posets of a type function}\label{app:structure}

In this section we collect some properties of the reduced structure poset $\Pe_f^0$ of a
type function $f$. 
The following lemma shows the basic  basic operations on $\Pe_f^0$, that are needed in the
sequel and in the main text.
The proof can be found in \cite[Lemma 4.18]{jencova2024onthestructure}. See also the appendices in
\cite{jencova2024onthestructure} for the necessary definitions.

\begin{lemma}\label{lemma:pf0_constr} Let $f\in \Fe_m$, $g\in \Fe_n$. 
\begin{enumerate}
\item[(i)] For $\sigma\in 
\permut_m$, $\Pe_{f\circ\sigma}^0\simeq \Pe_f^0$, with labels
$L_T^{f\circ\sigma}=\sigma^{-1}(L_T^f)$.
\item[(ii)] If $f$ or $f^*$ have free outputs, then $\Pe^0_{f^*}=\Pe^0_{f}\triangle
\{\emptyset,[m]\}$. Otherwise $\Pe_{f^*}^0=\Pe_f^0\triangle \{\emptyset\}$. 
\item[(iii)] Let $h= f\otimes g$, with the fixed decomposition  $[m+n]=[m]\oplus [n]$. Then
\[
\Pe_h^0\simeq \{(S,T)\colon S\in \Pe_{f}^0,\ T\in \Pe_g^0,\ S\in \Min(\Pe_f^0)\text{ or }T\in
\Min(\Pe_g^0)\},
\]
with labels 
\[
L^h_{(S,T)}=\begin{dcases} L^f_S & \text{ if } S\notin \Min(\Pe_f^0)\\
L^g_T+m & \text{ if } T\notin \Min(\Pe_g^0)\\
L^f_S\cup (L^g_T+m) & \text{ otherwise. }
\end{dcases}
\]
\item[(iv)] Let $h=f\vtl g$, with the fixed decomposition $[m+n]=[m]\oplus [n]$.  To
ease the description of the labels, let us first transform  the label sets in $\Pe_g^0$ as
 $L^g_T\mapsto L_T^g+m$. Then
\begin{itemize}
\item If $[m]\in \Pe_f^0$, then $\Pe_h^0\simeq(\Pe_f^0\setminus \{[m]\})\star \Pe_g^0$,
with the labels of $[m]$ added to the label sets of all elements in $\Min(\Pe_g^0)$.
\item If $[m]\notin \Pe_f^0$ and $\emptyset\in \Pe_g^0$, then $\Pe_h^0\equiv \Pe_f^0\star
(\Pe_g^0\setminus\{\emptyset\})$, the free outputs of $f$ added to the label sets  of all
elements in $\Min(\Pe_g^0\setminus\{\emptyset\})$.

\item If $[m]\notin \Pe_f^0$, $\emptyset\notin \Pe_g^0$ and $O_f^F\ne \emptyset$, then $\Pe_h^0\simeq
\Pe^0_f\star \{[m]\}\star \Pe_g^0$, with $L_{[m]}^h=O_f^F$.
\item If $[m]\notin \Pe_f^0$, $\emptyset\notin \Pe_g^0$ and $O_f^F= \emptyset$, then 
 $\Pe^0_h=\Pe_f^0\star \Pe_g^0$.
 \end{itemize}

\end{enumerate}

\end{lemma}

We will next prove some specific properties of the poset  $\Pe_f^0$.  These properties 
only pertain the order structure of the poset, and not the
labels. Since we have $\Pe^0_{f\circ\sigma}\simeq \Pe_f^0$ for any permutation $\sigma$, 
we may ignore possible permutations in the inductive construction of type functions. 
Consequently, in the proofs by induction below, it will be enough to show that the properties are true
for chains and preserved by complements and tensor products.

 For any element $S\in \Pe_f^0$, we will use the notations 
\[
 S^{\downarrow_f}:=\{X\in \Pe_f^0\colon X\le S\},\qquad
 S^{\uparrow_f}:=\{X\in \Pe_f^0\colon X\ge S\}.
 \]
  It is clear that $S^{\downarrow_f}$ is the downset generated by $S$,
 in $\Pe_f^0$, similarly, $S^{\uparrow_f}$ is the upset generated by $S$.

 \begin{lemma}\label{lemma:ypsilon} 
 Let $S,T\in \Pe_f^0$ be incomparable. Then   $S^{\downarrow_f}\cap T^{\downarrow_f} $ is either empty or
 a chain, similarly for  $S^{\uparrow_f} \cap T^{\uparrow_f}$.  
 \end{lemma}

\begin{proof} 
We will proceed by induction. The assertion is trivially true for any chain.
Assume that $f=f_1\otimes f_2$, then $S=(S_1,S_2)$ with  $S_i\in
\Pe_{f_i}^0$ and at least one of them is minimal, similarly for $T$. Assume that, say,
both $S_2$ and $T_2$ are minimal, then if $S^{\downarrow_f}\cap T^{\downarrow_f}$ is not empty, we must
have $S_2=T_2$ and $S^{\downarrow_f}\cap T^{\downarrow_f} =\{(V_1,S_2)\colon V_1\in
S_1^{\downarrow_{f_1}}\cap T_1^{\downarrow_{f_1}}\}$, which is a chain by the induction assumption. 

If $S_2$
and $T_1$ are minimal, then $V\le S$ implies  $V=(V_1,V_2)$
with $V_1\le S_1$ and $V_2=S_2$ and if $V\le T$, then $V_1=T_1$ and $V_2\le
T_2$. This implies that either $T_1\le S_1$, $S_2\le T_2$ and $V=(T_1,S_2)$ is the only element
in $S^{\downarrow_f}\cap T^{\downarrow_f}$, or the downset is empty. Similarly, $V\ge S$ implies that
$V_1\ge S_1$ and $V_2\ge S_2$, whereas $V\ge T$ implies $V_1\ge T_1$ and $V_2\ge T_2$.
Since at least one of $V_1, V_2$ must be minimal, we either have $V_2=S_2=T_2$ or
$V_1=S_1=T_1$, which leads
us to the first case. 
The proof
for $S^{\uparrow_f}\cap T^{\uparrow_f}$ is similar.

Assume now that the statement is
true for $f$ and let $S,T$ be incomparable elements in $\Pe_{f^*}^0$. Clearly, none  of $S$ and
$T$ can be the least or the largest element, so that $S$ and $T$ are also
incomparable elements of $\Pe_f^0$. Since the downsets/upsets in $\Pe_f^0$ differ from
those in  $\Pe_{f^*}^0$ only
by adding or removing  the least/largest  element, the statement follows by the assumption on
$f$.

\end{proof}
\begin{coro}\label{coro:xfree} 
 Let $S\in \Pe_f^0$. Then $S^{\downarrow_f}$ or $S^{\uparrow_f}$
is a chain. In the first case, the rank $\rho_f(S)$ of $S$ in $\Pe_f$ is the length of the
chain $S^{\downarrow_f}$.

\end{coro}

\begin{proof} Assume that $S^{\uparrow_f}$ is not a chain, then it must contain two
incomparable elements $U,V$. Since $\emptyset \ne S^{\downarrow_f}\subseteq
U^{\downarrow_f}\cap V^{\downarrow_f}$, the assertion follows by Lemma \ref{lemma:ypsilon}. 

To prove the second assertion, we once again proceed by induction. The statement is quite
clear for chains, note that in that case $\Pe_f=\Pe_f^0$.  Next, assume that $f=f_1\otimes f_2$, then we must have $S=(S_1,S_2)$
and if, say, $S_2$ is minimal, then $S^{\downarrow_f}=(S_1^{\downarrow_{f_1}},S_2)$. Hence
\[
\ell(S^{\downarrow_f})=\ell(S_1^{\downarrow_{f_1}})=\rho_{f_1}(S_1)+\rho_{f_2}(S_2)=\rho_f(S).
\]
Assume the statement is true for $f$ and let $S\in \Pe_{f^*}^0$ be such that
$S^{\downarrow_{f^*}}$ is a chain. If $S=[n]$, then this implies that $\Pe_{f^*}^0$ itself is a
chain. The statement also clearly holds if $S$ is the least element. Apart from these two
cases, we have  $S\in \Pe_f^0$ and $S^{\downarrow_f}$ differs from
$S^{\downarrow_{f^*}}$ just by adding/removing the least element. Since the same happens
with $\Pe_{f^*}$, we see that
\[
\ell(S^{\downarrow_{f^*}})=\ell(S^{\downarrow_f})\pm 1=\rho_f(S)\pm 1=\rho_{f^*}(S).
\]

\end{proof}

\begin{coro}\label{coro:vtl_type} Let $f,g$ be type functions. If  $f\vtl g$ is a type function, then   $f$
or $g$ must be  a chain. 

\end{coro}

\begin{proof} Assume that  $h=f\vtl g$ is a type function. Since the statement is trivial
if $f=1$ or $g=1$, we may assume that $\Pe_f^0, \Pe_g^0\ne \{\emptyset\}$. 

By Lemma
\ref{lemma:pf0_constr}, we see that as a  poset, $\Pe_h^0$ is isomorphic either to the
ordinal sum
$\Pe_h^0\simeq (\Pe_f^0\setminus \{[n]\})\star (\Pe_g^0\setminus \{\emptyset\})$, or there
is some element $X$ in between: $\Pe_h^0\simeq (\Pe_f^0\setminus \{[n]\})\star \{X\}\star
(\Pe_g^0\setminus \{\emptyset\})$. In either case, we have
\[
\Pe_f^0\setminus
\{[n]\}\subseteq \bigcap_{Y\in \Min(\Pe_g^0\setminus\emptyset)}Y^{\downarrow_{h}}.
\]
If $\Min(\Pe_g^0\setminus\emptyset)$ is not a singleton, Lemma \ref{lemma:ypsilon} implies
that $f$ must be a chain. Otherwise, $\Min(\Pe_g^0\setminus\emptyset)=:\{Y\}$ and we have  
\[
\Pe_f^0\setminus
\{[n]\}\subseteq Y^{\downarrow_{h}},\qquad \Pe_g^0\setminus \{\emptyset\}\subseteq
Y^{\uparrow_{h}}
\]
 Corollary \ref{coro:xfree} now implies that $f$ or $g$ must be a chain.

\end{proof}

\begin{coro}\label{coro:meet} Let $S,T\in \Pe_f^0$ be such that $S\wedge T$ exists in
$\Pe_f$. Then $S\wedge T\in \Pe_f^0$.

\end{coro}

\begin{proof} Assume that $X=S\wedge T\in \Pe_f$. If $X\notin \Pe_f^0$, there must be some
incomparable $X',X''$ below $X$ that are contained in $\Pe_f^0$. Indeed, any $i\in X$ must
be contained in some label  set of an element below $X$. Let $X'$ a maximal set such that
$i\in X'\le X$ and $X'\in \Pe_f^0$. Then clearly $X'\subsetneq X$, so that there must be
some  $j\in X\setminus X'$. Similarly as before, $j$  must be contained in some label set of
an element $X''\le X$. Since $X'$ is maximal, and does not contain $j$, it cannot be
comparable with $X''$. But $S\wedge T\notin \Pe_f^0$ also implies that $S$ and $T$ are
incomparable elements such that $X',X''\in S^{\downarrow_f}\cap T^{\downarrow_f}$. By Lemma
\ref{lemma:ypsilon}, this is impossible.

\end{proof}

\begin{coro}\label{coro:lattice} Let $S,T\in \Pe_f^0$. 
 $S$ and $T$ have a common lower bound if and only if $S\wedge T$ exists. $S$ and $T$ have
 a common upper bound if and only if  $S\vee T$ exists.  In particular, $\Pe_f^0$ is a lattice iff it is bounded.

\end{coro}

\begin{proof} Immediate from Lemma \ref{lemma:ypsilon}.

\end{proof}

\subsection{Normal form in some  examples} \label{app:normal}

In this paragraph, we derive a normal form for the examples in the main text. The simplest
nontrivial examples are the type functions of nonsignalling channels and process matrices in Examples \ref{exm:nspm}   and \ref{exm:pm_pf}.
In this case, normal forms are obtained straightforwardly from the definition of the type
function. We will demonstrate that these forms can be also seen from maximal chains in the
diagrams of the reduced structure posets.
For convenience, we redisplay the diagrams here:
\begin{center}
\vcbox{\includegraphics[scale=0.9]{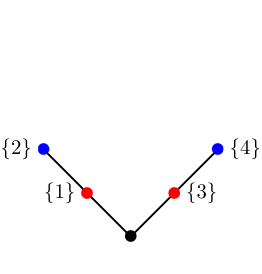}}
\vcbox{$\qquad,\qquad$}
\vcbox{\includegraphics[scale=0.9]{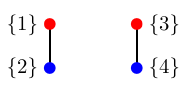}}
\vcbox{$\qquad,\qquad$}
\vcbox{\includegraphics[scale=0.9]{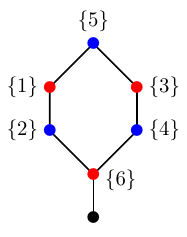}}
\end{center}
The leftmost diagram belongs to the type function $f_{ns}$ of nonsignalling channels, which
is the tensor product of two channel types. For the corresponding concatenation
$s=s^1s^2$, we will denote by $\gamma_2^i$ the channel type functions $2\to 1$ resp. $4\to
3$, with the upper
index indicating on which part of the string the respective function acts. We then have
$f_{ns}=\gamma_2^1\otimes \gamma_2^2$. The
normal form in this case is obtained from Lemma \ref{lemma:causalp} (iii) as
$f_{ns}=(\gamma_2^1\vtl \gamma_2^2)\wedge (\gamma_2^2\vtl\gamma_2^1)$. Notice that 
the diagram has two maximal chains, which we will
identify by the sequence of their labels as $\emptyset-1-2$ and $\emptyset-3-4$. Such
chains are precisely the structure posets of the channel types $\gamma^1_2$ and
$\gamma_2^2$.  These chains are connected at the least element, so there is no signalling
between their elements. 
The two concatenations of the chains in different orders, $\emptyset-1-2-3-4$ and
$\emptyset-3-4-1-2$, correspond exactly to 
the chain types appearing in the normal form above. The no signalling condition indicates
that we need to take the meet to neutralize signalling in the chains.

Similarly, the next diagram describing process matrices with type function $f_{pm}=(\tilde
\gamma_2^1\otimes \tilde \gamma_2^2)^*$. Here the normal form is $f_{pm}=(p_{\{2\}}\vtl
p_{\{4\}})\vee (p_{\{4\}}\vtl p_{\{2\}})$. The diagram has two maximal chains $2-1$ and
$4-3$ that are not connected, so that signalling is possible. The normal form is the join
of the two concatenations of the chains, that is $2-1-4$ (with free output 3) and $4-3-2$
(with free output 1). Here the possibility of signalling indicates that we should take the
join of the two chains.

The diagram on the right corresponds to process matrices with  global past and future.
The diagram has again two maximal chains: $\emptyset-6-2-1-5$ and $\emptyset-6-4-3-5$.
These chains are departing at an output (red) vertex, so that signalling is possible. The
chains also have a common vertex labeled by 5. We will concatenate the differing parts of
the chains, keeping the common parts below and above to retain the signalling conditions.
We obtain  the chain types $c_1\equiv \emptyset-6-2-1-4-3-5$ and
$c_2\equiv\emptyset-6-4-3-2-1-5$. Because of the signalling, we infer that the normal form is the join of these chain types.
This normal form can be obtained also from the expression $1_1\vtl f_{pm}\vtl p_1$ for this
type function.

We next turn to the more involved case of the adapters, see Example \ref{exm:adapters}.
Here we infer normal forms from the structure of the diagrams, and  confirm
by computation. Again, we redisplay the diagrams:
\begin{center}
\vcbox{\includegraphics[scale=0.9]{adapter.pdf}}
\vcbox{$\qquad$ $\qquad$}
\vcbox{\includegraphics[scale=0.9]{switch_to_switch.pdf}}

\end{center}
In the diagram for the type function $a_1$ on the left, there are two main branches
connected at the least element, so that there is no signalling between them. Each branch is
divided into 3 maximal chains, departing at red vertices labeled as 2 resp. 4, so that
signalling is possible between these chains in each branch. We index these chains as
$c_{i,j}$ where $i=1,2$ refers to the branch and $j=1,2,3$ to the individual chains. We
then compose them into larger chains as follows. For each $i$ and $j$, $h_{i,j}$ is a chain
starting as $c_{i,j}$ and appending the different elements from the other chains, first in
the $i$-th branch  (here
some freedom is possible) and then from the other branch, keeping the alternation of
inputs and outputs. We obtain the following chains
\begin{align*}
h_{1,1} &= \emptyset-2-6-5-8-7-1-4-3,\quad h_{2,1}=\emptyset-4-6-5-8-7-3-2-1 \\
h_{1,2}&=\emptyset-2-8-7-6-5-1-4-3,\quad  h_{2,2}=\emptyset-4-8-7-6-5-3-2-1\\
h_{1,3}&=\emptyset-2-1-4-6-5-8-7-3,\quad h_{2,3}=\emptyset-4-3-2-6-5-8-7-1
\end{align*}
Since these chains have even length, they  correspond to chain types, with the same inputs
and outputs as $a_1$. Putting joins and meets between the functions according to the
signalling conditions, we get
\[
a_1=(h_{1,1}\vee h_{1,2}\vee h_{1,3})\wedge(h_{2,1}\vee h_{2,2}\vee
h_{2,3})=(h_{1,1}\wedge h_{2,1})\vee (h_{1,2}\wedge h_{2,2})\vee (h_{1,3}\wedge h_{2,3}).
\]

The diagram for the type function $a_2$ on the right has three main branches departing at
a red vertex, hence signalling between branches is possible. The leftmost branch divides
into two maximal chains, departing at a blue vertex, indicating no signalling. These chains
are reconnected at the vertex labeled by 5. We create two longer chains by appending them,
but keeping the common parts below and above. The two chains are
\[
d_1=\emptyset-12-6-2-1-4-3-5-11,\qquad d_2=\emptyset-12-6-4-3-2-1-5-11.
\]
We now create the chain types for the normal form by appending the remaining maximal
chains with $d_1$ and $d_2$, again keeping the common parts above and below. In this way,
we create six chains $g_{i,j}$, where the  first index $i=1,2,3$ denotes the
the starting main branch and the second index $j=1,2$ indicates whether we used $d_1$ or
$d_2$. Again, there is some freedom in the order of the chains. We obtain
\begin{align*}
g_{1,1}&=\emptyset-12-6-2-1-4-3-5-8-7-10-9-11\\
g_{1,2}&=\emptyset-12-6-4-3-2-1-5-8-7-10-9-11\\
g_{2,1}&=\emptyset-12-8-7-10-9-6-2-1-4-3-5-11\\
g_{2,2}&=\emptyset-12-8-7-10-9-6-4-3-2-1-5-11\\
g_{3,1}&=\emptyset-12-10-9-8-7-6-2-1-4-3-5-11\\
g_{3,2}&=\emptyset-12-10-9-8-7-6-4-3-2-1-5-11.
\end{align*}
Taking into account the signalling conditions, we obtain
\[
a_2=(g_{1,1}\vee g_{2,1}\vee g_{3,1})\wedge (g_{1,2}\vee g_{2,2}\vee
g_{3,2})=(g_{1,1}\wedge g_{1,2})\vee (g_{2,1}\wedge g_{2,2})\vee (g_{3,1}\wedge g_{3,2}).
\]
Note that this normal form uses 6 chain types, which is more than the number of maximal
chains of $\Pe_{a_2}^0$, which is 4. Using the form $a_2=(k\otimes k^*)^*$ in Example
\ref{exm:adapters} and the strategy of the proof of Proposition
\ref{prop:normalform}, we can write the decomposition using only 4 chains, as
\[
a_2=(g_{1,1}\wedge g_{1,2})\vee k_2\vee k_3,
\]
where
\begin{align*}
k_2&=\emptyset-12-8-7-10-9-\{1,6\}-4-3-\{2,5\}-11\\
k_3&=\emptyset-12-10-9-8-7-\{1,6\}-4-3-\{2,5\}-11
\end{align*}
Here some indices have to be  grouped together into the same label sets.

Observe that all these normal forms apart from the last one satisfy the minimax condition
\eqref{eq:minimax}.

\end{document}